%% file: main.tex
\documentclass[prx,twocolumn,superscriptaddress]{revtex4-2}

\usepackage[usenames, dvipsnames]{xcolor}
\usepackage[utf8]{inputenc}
\usepackage{hyperref}
\hypersetup{colorlinks,allcolors=blue}
\usepackage{amsmath,revsymb,amssymb,mathtools}
\usepackage{tensor}
\usepackage{graphicx}
\usepackage{subcaption}
\usepackage{xcolor}
\usepackage{bbold,comment}
\usepackage{braket}
\usepackage{booktabs}
\usepackage{float}
\usepackage{physics, siunitx}
\usepackage{parskip}
\usepackage{multirow}
\usepackage{tabularx}
\usepackage{array}
\usepackage{pgf}
\usepackage{csquotes}
\usepackage{amsthm} 
\newcolumntype{C}[1]{>{\centering\let\newline\\\arraybackslash\hspace{0pt}}m{#1}}
\newcolumntype{Y}{>{\centering\arraybackslash}X}

\usepackage{color, colortbl}
\definecolor{Gray}{gray}{0.9}
\newcolumntype{g}{>{\columncolor{Gray}}l}

\renewcommand{\emph}{\textit}

\newcommand{\uli}{\overset{uli}{=}}
\newtheorem{proposition}{Proposition}
\newtheorem{result}{Result}
\newcommand{\aliceobs}[1]{A_{#1}}
\newcommand{\aliceproj}[2]{\Lambda^{#1}_{#2}}
\newcommand{\aliceseq}[2]{{\mathsf A}^{#1}_{#2}}
\newcommand{\bobobs}[1]{B_{#1}}
\newcommand{\bobproj}[2]{\Pi^{#1}_{#2}}
\newcommand{\bobseq}[2]{{\mathsf B}^{#1}_{#2}}
\newcommand{\eveproj}[1]{E_{#1}}
\newcommand{\alicekraus}[2]{K^{#1}_{#2}}
\newcommand{\bobkraus}[2]{K^{#1}_{#2}}

\begin{document}
    
\title{Secure and robust randomness with sequential quantum measurements}

\author{Matteo Padovan}
\affiliation{Dipartimento di Ingegneria dell'Informazione, Universit\`a degli Studi di Padova, via Gradenigo 6B, IT-35131 Padova, Italy}
\affiliation{Centro di Ateneo di Studi e Attivit\`a Spaziali ``Giuseppe Colombo", Universit\`a di Padova, via Venezia 15, IT-35131 Padova, Italy}

\author{Giulio Foletto}
\thanks{Affiliated to KTH Royal Institute of Technology, Stockholm SE-106 91, Sweden at the time of revision and publication.}
\affiliation{Dipartimento di Ingegneria dell'Informazione, Universit\`a degli Studi di Padova, via Gradenigo 6B, IT-35131 Padova, Italy}

\author{Lorenzo Coccia}
\affiliation{Dipartimento di Ingegneria dell'Informazione, Universit\`a degli Studi di Padova, via Gradenigo 6B, IT-35131 Padova, Italy}

\author{Marco Avesani}
\affiliation{Dipartimento di Ingegneria dell'Informazione, Universit\`a degli Studi di Padova, via Gradenigo 6B, IT-35131 Padova, Italy}

\author{Paolo Villoresi}
\affiliation{Dipartimento di Ingegneria dell'Informazione, Universit\`a degli Studi di Padova, via Gradenigo 6B, IT-35131 Padova, Italy}
\affiliation{Padua Quantum Technologies Research Center, Universit\`a degli Studi di Padova, via Gradenigo 6B, IT-35131 Padova, Italy}

\author{Giuseppe Vallone}
\email{vallone@dei.unipd.it}
\affiliation{Dipartimento di Ingegneria dell'Informazione, Universit\`a degli Studi di Padova, via Gradenigo 6B, IT-35131 Padova, Italy}
\affiliation{Padua Quantum Technologies Research Center, Universit\`a degli Studi di Padova, via Gradenigo 6B, IT-35131 Padova, Italy}
\affiliation{Dipartimento di Fisica e Astronomia, Universit\`a degli Studi di Padova, via Marzolo 8, IT-35131 Padova, Italy}

\begin{abstract}
    Quantum correlations between measurements of separated observers are crucial for applications like randomness generation and key distribution. 
    Although device-independent security can be certified with minimal assumptions, current protocols have limited performances. 
    Here, we exploit sequential measurements, defined with a precise temporal order, to enhance performances by reusing quantum states. 
    We provide a geometric perspective and a general mathematical framework, analytically proving a Tsirelson-like boundary for sequential quantum correlations, which represents a trade-off in nonlocality shared by sequential users. 
    This boundary is advantageous for secure quantum randomness generation, certifying maximum bits per state with one remote and two sequential parties, even if one sequential user shares no nonlocality. 
    Our simple qubit protocol reaches this boundary, and numerical analysis shows improved robustness under realistic noise. 
    A photonic implementation confirms feasibility and robustness. 
    This study advances understanding of sequential quantum correlations and offers insights for efficient device-independent protocols.
\end{abstract}

\maketitle

\section{Introduction}
The effectiveness of information security protocols, whether quantum or classical, relies on specific assumptions.
Classical protocols typically make considerations about the computational capabilities of adversaries. 
On the contrary, the security of quantum protocols is based solely on the validity of quantum theory. 
However, to leverage this validity in practice, certain assumptions about the implementation are necessary, making the protocol device-dependent.
Efforts towards minimizing assumptions for enhanced security evaluation give rise to the device-independent approach in quantum information. 
A protocol is deemed device-independent when its security remains guaranteed without assumptions about the internal workings of the devices used in its implementation. 
In these schemes, a physical system prepared in an entangled state is shared and measured by different users, who choose their measurements randomly.
Entanglement is necessary to produce correlations that are not reproducible by any local hidden variable theory. 
These correlations are referred to as nonlocal.
The outcomes serve the dual purpose of manifesting nonlocality and providing a useful classical resource, such as a key or random bit.
In principle, the security of this resource is guaranteed by nonlocality even if the devices implementing the protocols are entirely untrusted or controlled by adversaries.

A major drawback of these schemes is the low rate of resource extraction. 
This is mainly due to the challenges of creating and preserving entanglement, which is degraded by the coupling of the system with the environment.
Instead of relying on faster entanglement generation, which may be feasible in the future, we study how to optimize the extraction of useful resources from each single entangle system.
A way of doing so proposed in the scientific literature uses weak measurements to realize sequential protocols, in which each quantum system is measured more times \cite{Silva2015,Mal2016,Schiavon2016,Curchod2017a,Hu2018,Tavakoli2018a,Brown2020,Foletto2020,Foletto2021}.
Often, they are direct extensions of schemes that use projective measurements, adding further intermediate measurements, and improving the performance in terms of resources extracted from the same quantum system.
With the strategy proposed in \cite{Curchod2017a} it is even possible, in principle, to produce an unlimited amount of device-independent randomness for each generated bipartite entangled state.
However, the robustness to noise of this protocol is limited and therefore requires great accuracy of realization \cite{Foletto2021}.

At the same time, the appeal of sequential protocols lies in the correlations they can create, for instance for the possibility of sharing nonlocality among multiple users \cite{Cheng2021, Cheng2022, Steffinlongo2022}. 
An effective approach for characterizing quantum correlations is through a geometrical perspective.
However, although the geometry of quantum correlations has been the subject of several studies \cite{Brunner2014,Christensen2015,Goh2018}, its extension to the sequential setting is little known.
Most previous analyses focus only on the correlations between each sequential user and the remote one, finding a monogamy trade-off: stronger correlations for one user imply weaker ones for the others \cite{Silva2015,Curchod2017a,Foletto2020, Brown2020}.
A detailed investigation of the trade-off, its geometry, and its implications could help formulate better quantum protocols that could overcome this compromise \cite{Gallego2014,Bowles2020}.
Moreover, the literature lacks a general mathematical framework that is useful for characterizing sequential quantum correlations.

In this paper, we characterize sequential quantum correlations with a geometric approach.
First, we provide a general mathematical framework useful for describing any sequential quantum scenario.
Then, we extend the common two-user, two-measurement, two-outcome scenario with a sequential user on one side and study the geometry of the obtainable correlations, identifying a Tsirelson-like quantum boundary that also serves as monogamy trade-off. 
This trade-off provides further insights about the sharing of nonlocality between sequential users.

Furthermore, we show that the correlations on the boundary can be used to certify the maximum amount of local randomness obtainable for our scenario, that is, two bits.
This is possible regardless of how nonlocality, quantified as violation of a given Bell inequality, is divided between the sequential pairs, meaning that the trade-off for nonlocality is not a trade-off for randomness.
Contrary to intuition, the maximal number of bits is attained even if the correlations generated by one of the pairs are entirely local.
This is in contrast to previous results in which randomness was generated from nonlocal pairwise correlations \cite{Curchod2017a}, and offers a perspective for future works.
We also propose an explicit protocol that can generate boundary correlations using states and measurements similar to those that maximally violate the CHSH inequality.
Compared to the protocol of Ref.\ \cite{Curchod2017a}, which can also achieve two bits of randomness with two dichotomic measurements, ours is simpler as it requires fewer different settings.
Furthermore, we show numerically through semidefinite programming techniques that it is more robust to noise, because it is insensitive to the nonlocality trade-off.
Our protocol is also simpler than the two proposals of \cite{Acin2016}, which can also certify two bits of randomness, since it requires fewer measurements, allowing for easier experimental implementation.

Finally, to demonstrate the feasibility and noise resilience of our protocol, we performed a proof-of-concept experimental test based on polarization-entangled photon pairs that generate the correlations required by the protocol.
From these correlations, we could certify 39\% more random bits than those obtainable with standard non-sequential CHSH protocols in the same noise conditions.

The paper is structured as follows.
In Sec.\ \ref{sec:seq_scenario} we introduce a convenient formalism to describe the sequential quantum scenario. 
In Sect. \ref{sec:bounds} we show some inequalities in the values of the correlations.
In Sec.\ \ref{sec:seqCHSHprotocol} we propose a protocol to saturate these inequalities: This will lead us to identify part of the boundary of the sequential quantum correlation.
In Sec.\ \ref{sec:randomness}, we show how the saturation of the inequalities can be used to certify randomness, supporting the discussion with numerical simulations for some non-ideal cases.
Finally, in Sec.\ \ref{sec:experiment} we describe a proof-of-concept sequential quantum experiment based on quantum optics.

\section{Methods}
\subsection{The sequential scenario}

We work in the sequential scenario defined in Ref. \cite{Gallego2014}, and specifically in a scenario that includes three users: Alice, Bob$_1$, Bob$_2$.
A schematic is depicted in Fig.\ \ref{fig:boxes}.
A common source prepares an unknown physical system that is shared and then measured by the untrusted devices operated by the three users.
Each user randomly chooses a measurement identified by a binary input $x, y_1, y_2 \in \{0,1\}$ and obtains as a result a binary output $a, b_1, b_2 \in \{\pm 1\}$.
We assume all inputs to be independent of one another, and forbid any communication during data collection between Alice and the Bobs, but we allow unidirectional communication from Bob$_1$ to Bob$_2$ between the production of their respective outputs: This characterizes the sequential correlation scenario that we formally define below.

The main goal of this work is to study the properties of the correlations between inputs and outputs $p(a, \mathbf{b} | x, \mathbf{y}) = p(a,b_1,b_2|x,y_1,y_2)$ that can be generated in this scenario and to see how they can be used to produce device-independent random numbers from the Bobs outputs.
We assume that after sufficiently many independent and identically distributed runs, the correlations are known perfectly, neglecting the effects of finite statistics.
Moreover, we do not have requirements on the probabilities of the inputs, as long as they allow for the entire reconstruction of the correlations $p(a, \mathbf{b} | x, \mathbf{y})$.

The absence of communication means that Alice's marginal probabilities are independent of Bobs' inputs and vice versa.
Formally, the correlations must satisfy the no-signaling conditions \cite{Brunner2014}:

\begin{equation}
    \begin{split}
        \sum_a p(a, \mathbf{b} | x, \mathbf{y}) &= \sum_a p(a, \mathbf{b} | x', \mathbf{y}) \quad \forall \mathbf{b}, x, x', \mathbf{y} \\
        \sum_{\mathbf{b}} p(a, \mathbf{b} | x, \mathbf{y}) &= \sum_{\mathbf{b}} p(a, \mathbf{b} | x, \mathbf{y'}) \quad \forall a, x, \mathbf{y}, \mathbf{y'}
    \end{split}
\end{equation}
Furthermore, sequentiality implies that Bob$_2$'s input cannot influence Bob$_1$ \cite{Gallego2014}:

\begin{multline}
\label{eq:gallego_seq_cons}
    \sum_{b_2} p(a, b_1, b_2| x, y_1, y_2) = \sum_{b_2} p(a, b_1, b_2| x, y_1, y'_2) \\ \forall a, b_1, x, y_1, y_2, y'_2
\end{multline}

As is common in the context of device-independent protocols, we focus on the set of sequential quantum correlations $Q_{SEQ}$, i.e.\  those sequential correlations that can be written using Born rule as

\begin{multline}
    \label{eq:prob_quant_seq}
    p(a,\mathbf{b} | x, \mathbf{y}) =   \\
    \sum_{\mu, \mu_1,\mu_2} \hspace{-1.5ex} \mathrm{Tr} \Bigl[(\alicekraus{x}{a,\mu}\otimes \bobkraus{y_2}{b_2,\mu_2}\bobkraus{y_1}{b_1,\mu_1})\rho 
    {(\alicekraus{x}{a, \mu} \otimes {\bobkraus{y_2}{b_2,\mu_2}}\bobkraus{y_1}{b_1,\mu_1})}^\dag \Bigr]
\end{multline}
where we have imposed the standard tensor product form to separate Alice and the Bobs \cite{Navascues2012}, and the measurements are described in terms of Kraus operators such that $\sum_{a, \mu} {\alicekraus{x}{a, \mu}}^\dagger \alicekraus{x}{a, \mu} = \sum_{b_1, \mu_1} {\bobkraus{y_1}{b_1,\mu_1}}^\dagger \bobkraus{y_1}{b_1,\mu_1} = \sum_{b_2,\mu_2} {\bobkraus{y_2}{b_2,\mu_2}}^\dagger \bobkraus{y_2}{b_2,\mu_2} = \openone$ for any input \cite{Bowles2020}.
We note that the sequentiality is reflected in the order of the Kraus operators.

Expression \eqref{eq:prob_quant_seq} considers that a real implementation of a protocol can be generally described in terms of mixed states and non-projective measurements.
However, the shared state can be assumed to be pure because even if the actual is not, it is always possible to consider its purification in a larger Hilbert space without changing the correlations.
Similarly, through the Stinespring \cite{Stinespring1955} dilation, the right-hand side of \eqref{eq:prob_quant_seq} can be rewritten in terms of orthogonal projective measurements satisfying additional constraints in order to guarantee \eqref{eq:gallego_seq_cons}
(see Supplementary Methods A and \cite{Neumark1943, Bowles2020} for details).

The description in terms of pure states and projective measurements is more convenient for studying the geometry of sequential quantum correlations and we will adopt it in the following.
Moreover, as shown in Supplementary Methods A, it can also be rephrased in terms of 
unitary and hermitian operators (namely measurement operators with only $\pm1$ eigenvalues) $\bobobs{y_1}$ and $\bobobs{y_1,y_2}$  that satisfy the following constraints:
\begin{equation}
    \begin{split}
        &[\bobobs{y_1},\bobobs{y_1,y_2}] = 0 \\
        &\bobobs{y_1}^\dagger =  \bobobs{y_1} \ , \quad \bobobs{y_1,y_2}^\dagger =  \bobobs{y_1,y_2} \hspace{1cm} \forall y_1,y_2  \\
        &\bobobs{y_1}^\dagger \bobobs{y_1} = \bobobs{y_1,y_2}^\dagger \bobobs{y_1,y_2} = \openone  \,.
    \end{split}
    \label{eq:C_constraints}
\end{equation}
They can be understood as the observables measured by Bob$_1$ and Bob$_2$. Indeed, the operators $\bobobs{y_1}$ reproduce the statistics of Bob$_1$, while the operators $\bobobs{y_1,y_2}$ reproduce the statistics of Bob$_2$ given that Bob$_1$ has chosen the input $y_1$. 
See Fig.\ \ref{fig:boxes} for a schematic of the scenario.
Similar considerations can be applied to Alice's side to define two unitary and hermitian operators $\aliceobs{x}$.
Without the sequentiality requirement, Alice has no commutation relation analogue to that of \eqref{eq:C_constraints}.

With these definitions, the correlations in the sequential scenario can always be written as
\begin{equation}\label{eq:prob_proj}
    p(a,\mathbf{b} | x, \mathbf{y}) = \expval*{\aliceproj{x}{a}\otimes \bobproj{y_1}{b_1}\bobproj{y_1,y_2}{b_2}}{\psi}
\end{equation}
where $\aliceproj{x}{a}$, $\bobproj{y_1}{b_1}$ and $\bobproj{y_1,y_2}{b_2}$ are the projectors on the eigenspaces of hermitian operators $\aliceobs{x}$, $\bobobs{y_1}$ and $\bobobs{y_1,y_2}$ respectively (i.e., $\bobobs{y_1}=\bobproj{y_1}{+}-\bobproj{y_1}{-}$ and similarly for Alice and Bob$_2$).
The commutation relation in \eqref{eq:C_constraints} guarantees that the product of the Bobs' projectors can be used to compute a well-defined probability.

\label{sec:seq_scenario}
\begin{figure}
    \centering
    \includegraphics[width=\linewidth]{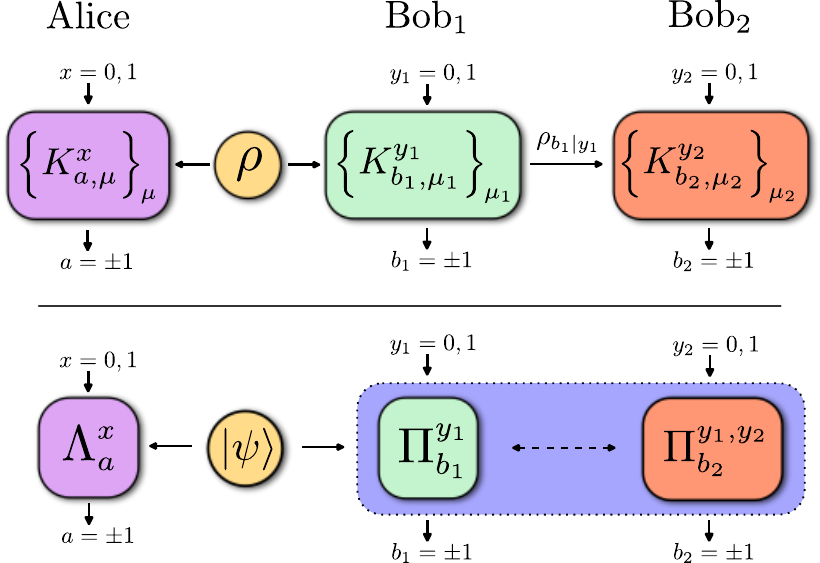}
    \caption{Schematic of the sequential scenario. Above, the framework with Kraus operators. Bottom, the projective framework with the operators introduced in Sec. \ref{sec:seq_scenario}.}
    \label{fig:boxes}
\end{figure}

\section{Results}

\subsection{Bounds on the sequential quantum correlations}
\label{sec:bounds}
Having introduced the notation, we now present a Tsirelson-like bound satisfied by the sequential quantum correlations.

In our specific case, we are not interested in the operations of Bob$_2$ after Bob$_1$ has chosen the input $y_1=1$.
This means that we will only consider the marginal probability distribution $p(a, b_1 |x, y_1=1) = \sum_{b_2} p(a,b_1,b_2|x,y_1=1,y_2)$.
We are allowed to do this because of sequentiality: Bob$_2$ cannot influence Bob$_1$ and hence $p(a, b_1 |x, y_1=1)$ is well defined and does not depend on $y_2$.
Our results are valid regardless of what Bob$_2$ does in this case, we can even think that he does not perform any measurement at all.
With this simplification, the association of inputs and measurements is as follows:
\begin{equation}
\begin{split}
    \text{\textit{Input sequence}} &\hspace{0.75cm} \text{\textit{Measurements}}\\
    y_1,y_2=0,0 &\hspace{0.75cm} \bobobs{0} \text{ and } \bobobs{0,0} \\ 
    y_1,y_2=0,1 &\hspace{0.75cm} \bobobs{0} \text{ and } \bobobs{0,1} \\
    y_1, y_2=1,0 \text{ or } 1,1 &\hspace{0.75cm} \bobobs{1}\,.
\end{split}
\end{equation}
Consequently, we consider the following operators:
\begin{equation}
    \begin{split}
        S_1&\equiv(\aliceobs{0}
    + \aliceobs{1})\bobobs{0} +(\aliceobs{0} - \aliceobs{1})\bobobs{1}\\
        S_2&\equiv(\aliceobs{0} + \aliceobs{1})\bobobs{0,0}+(\aliceobs{0} - \aliceobs{1})\bobobs{0,1}\,.
    \end{split}
\end{equation}
These are two CHSH-like operators relative to Alice-Bob$_1$ and Alice-Bob$_2$ respectively, and their mean values can be measured in our scenario from the correlations $p(a,\mathbf{b}|x, \mathbf{y})$ by selecting the values of the inputs that correspond to the relevant observables.
Hence, the usual results about CHSH operators also apply, so that in a quantum setting $\expval*{S_i}\leq 2\sqrt{2}$.
Moreover, from conceptually similar results in the literature \cite{Silva2015,Curchod2017a,Foletto2020}, one can expect a trade-off between $\expval*{S_1}$ and $\expval*{S_2}$, therefore it is meaningful to consider an expression that combines the two:
\begin{equation}
    \label{eq:Stheta}
        S_\theta\equiv \cos2\theta (S_1-\sqrt{2}\openone) + \sin 2\theta (S_2-\sqrt{2}\openone)\,.
\end{equation}
Furthermore, we introduce the operator
\begin{equation}
    S_c\equiv(\aliceobs{0} + \aliceobs{1})\bobobs{0,0} +(\aliceobs{0} - \aliceobs{1})\bobobs{1}
\end{equation}
whose expected value is a function of part of the statistic of Alice-Bob$_1$ and part of the statistic of Alice-Bob$_2$.
This is a well-defined CHSH-like operator, as the relevant observables on the Bobs' side are measured with different inputs: $y_1, y_2 = 1,0 \text{ or } 1,1$ for $\bobobs{1}$ and $y_1, y_2 = 0,0$ for $\bobobs{0,0}$.
Therefore, in a quantum experiment $\expval*{S_c} \leq 2\sqrt{2}$.

We can express now our main result (proven in Supplementary Methods B) on the geometry of the sequential correlations, which is a bound on $\expval*{S_1}$ and $\expval*{S_2}$ in the specific case in which $\expval*{S_c}$ takes its maximum value $2\sqrt{2}$.
\begin{result}
\label{res:geometry}
    For any sequential quantum correlation in our scenario, it holds that
    \begin{equation}
    \label{eq:bellineqseq}
        \expval*{S_c}=2\sqrt{2}\quad\Rightarrow\quad \expval*{S_\theta} \leq \sqrt{2}\,,\quad \forall \theta\,,
    \end{equation}
    and there exist correlations that saturate the inequality.
\end{result}
This upper bound on $\expval*{S_\theta}$ can be interpreted as a monogamy relation between the correlations of Alice-Bob$_1$ and Alice-Bob$_2$.
This is different from the trade-offs already present in the literature because $S_2$ considers Bob$_1$'s input, since $\bobobs{0,0}$ and $\bobobs{0,1}$ are measured only if $y_1 = 0$.
Instead, in Ref.\ \cite{Silva2015}, the quantity similar to $S_2$ is calculated ignoring the actions of Bob$_1$, while the protocols of Refs.\ \cite{Curchod2017a,Foletto2020} calculate separate CHSH quantities for each of Bob$_1$'s outputs, adapting Alice's measurements to obtain the highest values.

\begin{figure}
    \centering
    \includegraphics[width=0.9\linewidth]{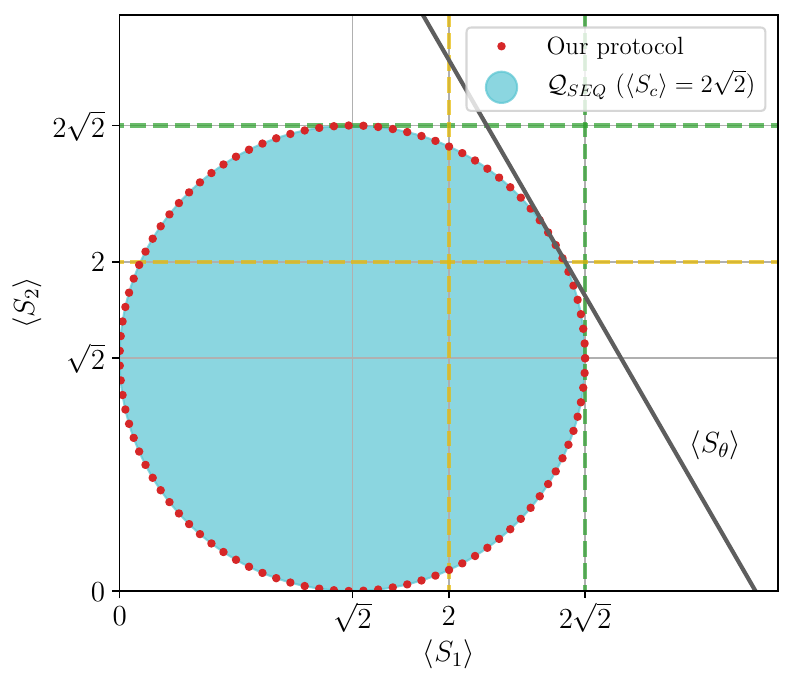}
    \caption{Cross-section of the set of sequential quantum correlations, $\mathcal{Q}_{SEQ}$, at $\expval{S_c}=2\sqrt{2}$.
    The dashed lines denote  the maximum values achievable by $\expval{S_1}$ and $\expval{S_2}$ in the local and non-sequential quantum scenarios, without restrictions on $\expval{S_c}$. The red dots mark the geometric location of the correlations achievable with the protocol explained in Sec \ref{sec:seqCHSHprotocol}.}
    \label{fig:boundary}
\end{figure}

\subsection{Sequential-CHSH protocol}
\label{sec:seqCHSHprotocol}
In the following we will provide, for any given value of $\theta$, state and operators that generate correlations for which $\expval*{S_c}=2\sqrt{2}$ and $S_\theta = \sqrt{2}$, proving that the inequality \eqref{eq:bellineqseq} is tight and identifies a boundary of $Q_{SEQ}$ in our scenario. 

In the scheme, Alice and $\text{Bob}_1$ share the maximally entangled Bell state $\ket*{\phi^+}_{AB} = \qty(\ket{00} + \ket{11})/\sqrt{2}$, where $\ket{0}$ and $\ket{1}$ are the eigenstates of the $\sigma_z$ Pauli matrix. 

Alice randomly chooses between two inputs $x \in \{0,1\}$, corresponding to the two observables
\begin{equation}
    \aliceobs{0} = \frac{\sigma_z + \sigma_x}{\sqrt{2}}\hspace{1cm}
    \aliceobs{1} = \frac{\sigma_z - \sigma_x}{\sqrt{2}} \,.
\end{equation}

$\text{Bob}_1$ randomly chooses between two inputs $y_1 \in \{0,1\}$, the latter corresponding to a projective measurement of $\sigma_x$ and the former to the non-projective measurement realized by the two Kraus operators depending on the parameter $\theta$:
\begin{equation}
\label{eq:bob1_kraus}
    \begin{aligned}
    K_{+}\qty(\theta) &= \cos\theta\dyad{0} + \sin\theta\dyad{1} \,, \\
    K_{-}\qty(\theta) &= \cos\theta\dyad{1} + \sin\theta\dyad{0} \,.
    \end{aligned} 
\end{equation}
In this expression, the parameter $\theta$ is taken to be the same as the one appearing in \eqref{eq:Stheta} in order to achieve correlations that saturate the inequality in \eqref{eq:bellineqseq}. Its value has a clear physical meaning: it controls the strength of the measurement, in the sense that $\theta=n\frac{\pi}{2}$ leads to a projective measurement of $\pm \sigma_z$, while for $\theta=\frac{\pi}{4}+n\pi$ correspond to a non-interactive measurement. At $\theta=\frac{\pi}{4}+n\frac{\pi}2$ the two Kraus operators are equal, up to a sign.

After these operations, if $y_1=1$, the protocol ends.
Otherwise, for $y_1=0$, $\text{Bob}_1$ sends the post-measurement state to $\text{Bob}_2$, who randomly chooses between the projective measurements of $\sigma_z$ or $\sigma_x$, each corresponding to one of the two inputs $y_2 \in \{0,1\}$.

As discussed in Supplementary Methods D, in terms of projective operators, this protocol can be formulated by leaving unchanged $\aliceobs{0}$ and $\aliceobs{1}$, while introducing the operators
\begin{equation}
    \begin{split}
        \bobobs{0} &= \sigma_z \otimes \sigma_z \\
        \bobobs{1} &= \sigma_x \otimes \openone_{B''} \\
        \bobobs{0,0} &= \sigma_z \otimes \openone_{B''} \\
        \bobobs{0,1} &= \sigma_x \otimes \sigma_x
    \end{split}
    \label{eq:our_O_operators}
\end{equation}
on the Bobs' side.
These act on an Hilbert space $\mathcal{H}_{B'}\otimes\mathcal{H}_{B''} = \mathbb{C}^2\otimes\mathbb{C}^2$.
The shared state is now
\begin{equation}
    \ket{\psi} = \ket*{\phi^+}_{AB'}\Bigl[\cos\theta \ket{0}_{B''} + \sin\theta \ket{1}_{B''}\Bigr]
\end{equation}
One can verify, using Eqs.\ \eqref{eq:prob_quant_seq} and \eqref{eq:prob_proj}, that the sequential and projective formulations give the same correlations, and that the operators $\bobobs{y_1}$ and $\bobobs{y_1,y_2}$ respect all the constraints in Eq.\ \eqref{eq:C_constraints}.
Moreover the relations $\expval*{S_c} = 2\sqrt{2}$ and $\expval*{S_\theta}=  \sqrt{2}$ hold
with the above defined state and operators, 
proving that the inequality on $S_\theta$ is tight and define a boundary, as claimed.

A geometrical depiction of this boundary is shown in  Fig.\ \ref{fig:boundary} and can be deduced by Eq.\ \eqref{eq:Stheta}: For each $\theta$, when $\expval{S_\theta}=\sqrt{2}$, this equation describe the tangent to a circumference in the $\expval{S_1} \expval{S_2}$ plane, centered at $(\sqrt{2},\sqrt{2})$ and of radius $\sqrt{2}$. 
The points on the circumference are spanned by the protocol just discussed, while the interior of the circle is filled with sequential quantum correlations satisfying $\expval{S_c}=2\sqrt{2}$ and $\expval{S_\theta}<\sqrt{2}$.

\subsection{Randomness from correlations}
\label{sec:randomness}
We can now move to our second main result, which is a statement on the randomness that can be obtained from correlations on the aforementioned boundary of $Q_{SEQ}$.
In this work, we consider only local randomness, originating solely from the side of the Bobs.
Given a sequential probability distribution that is observed experimentally $P_{exp}(a,\mathbf{b}|x,\mathbf{y})$, the quantity of device-independent random numbers that can be extracted from the outcomes corresponding to a specific input sequence $\mathbf{y_r}$ can be measured by the (quantum conditional) min-entropy $H_{\text{min}}=-\log_2 G$ \cite{Brown2020entropy}, where $G$ is the maximum guessing probability that an adversary Eve has on the Bobs' outcomes when the input sequence is $\mathbf{y_r}$: 

\begin{gather}
\label{eq:SDP_problem}
    G = \max_{p_{ABE}} \sum_{\mathbf{b}} p_{BE} \qty( \mathbf{b},\mathbf{b}|\mathbf{y_r} ) \\
\begin{aligned}
\label{eq:SDP_constraints}
    \text{s.t.} \hspace{0.5cm} \sum_{\mathbf{e}} p_{ABE}(a,\mathbf{b},\mathbf{e}|x,\mathbf{y}) &= P_{\rm exp}(a,\mathbf{b}|x,\mathbf{y})\, , \\
    p_{ABE}(a,\mathbf{b},\mathbf{e}|x,\mathbf{y}) &\in \mathcal{Q}_{SEQ}\,.
\end{aligned}
\end{gather}
The first condition of Eq.\ \eqref{eq:SDP_constraints} compels Eve to use a strategy $p_{ABE}$ that is compatible with the experimental correlations $P_{exp}(a,\mathbf{b}|x,\mathbf{y})$.
The second means that the strategy is also quantum in the sense explained in Sec.\ \ref{sec:seq_scenario} and the sequentiality requirement applies only to the Bobs.

With this definition, we can express the second main result of our work:
\begin{result}
\label{res:randomness}
    For any sequential quantum correlation in our scenario such that $\expval*{S_c} = 2\sqrt{2}$ and $\expval*{S_\theta} = \sqrt{2}$ for a given $\theta \neq n \frac{\pi}{4}$,
    the min-entropy is 
    \begin{equation}
        H_{\mathrm{min}}=2 \ \text{bits} 
    \end{equation}
    when evaluated with the input sequence $\mathbf{y_r}=(0,1)$. If $\expval*{S_\theta} = \sqrt{2}$ for some $\theta=n\frac{\pi}{4}$, it reduces to $H_{\mathrm{min}} = 1$ bit.
\end{result}
The proof, provided in Supplementary Methods C, 
is based on the self-testing properties of the CHSH inequality \cite{Supic2020}, which are valid because $\expval*{S_c} = 2\sqrt{2}$, and on the additional necessary conditions that the quantum state and measurements must satisfy in order to saturate also Eq.\ \eqref{eq:bellineqseq}.
We emphasize that the demonstration is conducted in a device-independent scenario, and it remains valid regardless of the dimension and specific details of the sequential quantum realization. Examples of states and operators capable of producing 2 bits of randomness are those described in Section \ref{sec:seqCHSHprotocol}.

Two dichotomic measurements can provide at most two random bits.
The fact that they achieve this bound, certifies the complete unpredictability of their outcomes.
This descends from the features of the entire correlation $P_{exp}(a,\mathbf{b}|x,\mathbf{y})$ and not just from the pairwise ones.
Indeed $\expval*{S_1}$ and $\expval*{S_2}$ cannot be maximized simultaneously, and the situations in which one is maximized are exactly those for which the randomness drops to one bit.
By compromising on their respective nonlocality, Bob$_1$ and Bob$_2$ achieve the best results in terms of randomness.
There are even regions on the boundary in which either the correlations between Alice and Bob$_1$ or those between Alice and Bob$_2$ are entirely local, as can be checked by verifying that all CHSH inequalities involving their paired results are respected.
Yet, thanks to the three-party correlations, the min-entropy is still maximal at two bits.

However, due to unavoidable experimental imperfections, a real implementation cannot generate ideal correlations that sit exactly at the boundary, therefore it is important to study the amount of device-independent randomness in the interior of $Q_{SEQ}$.
We address this problem numerically using the Navascu\'{e}s-Pironio-Ac\'{i}n (NPA) hierarchy \cite{Navascues2007,Navascues2008}, and its sequential generalization \cite{Bowles2020}.
This tool replaces the usually difficult-to-verify second condition in \eqref{eq:SDP_constraints} with an ordered series of increasingly stringent necessary conditions on linear combinations of the probabilities $p_{ABE}(a,\mathbf{b},\mathbf{e}|x,\mathbf{y})$.
The constraint $p_{ABE}(a,\mathbf{b},\mathbf{e}|x,\mathbf{y}) \in \mathcal{Q}_{SEQ}$ is retrieved when all conditions are satisfied, but stopping to a finite order $k$ of the series allows casting the problem to a practical semi-definite program (SDP) \cite{Boyd2004} and restricts $p_{ABE}$ to belong to a set $\mathcal{Q}_{SEQ}^k \supseteq \mathcal{Q}_{SEQ}$ \cite{Bowles2020}.
This means that the optimization is performed over a larger set of correlations than what is allowed by quantum mechanics and gives Eve more power than she actually has.
The solution of the program is then an upper bound of the actual guessing probability: Finding a value $G$ through the SDP certifies in a device-independent way that the min-entropy of the two outcomes is at least $-\log_2G$ bits.

Numerical issues could in principle overestimate the min-entropy, but this can be prevented by giving tolerances to the constraints of Eq.\ \eqref{eq:SDP_constraints}.
These tolerances always benefit Eve and, if chosen much larger than the machine precision, overwhelm its the potentially dangerous effect \cite{Winick2018}.

Rather than computing the min-entropy for all possible values of $\expval*{S_c}$ and $\expval*{S_\theta}$, we do it in the context of the protocol explained in Sec.\ \ref{sec:seqCHSHprotocol}, so as to study also its noise robustness.
We numerically generate the experimental correlations using the maximally entangled state $\ket*{\phi^+}$ mixed with random noise, namely $\rho_{AB}=(1-p)\dyad{\phi^+}+p\openone/4$, and the measurements required by the protocol.
We then set these correlations as constraint in the optimization problem \eqref{eq:SDP_problem}. 
We perform such computation for different values of the strength parameter $\theta$, since, for noisy states, different values of $\theta$ could influence the performance of the protocol by imposing different limitations on Eve's strategies.
Because of the symmetry of the protocol, it is sufficient to restrict the analysis to $\theta \in [0,\frac{\pi}{4}]$.
For such numerical computations we adopt Ncpol2sdpa \cite{Wittek2015} and the solver SDPA-DD \cite{Nakata2010}, setting a minimal solver precision of $10^{-12}$ for all the theoretical simulations.
The NPA order is $1+\text{AB}$ \cite{Bowles2020}: this order is enough for retrieving the analytical result in the ideal case scenario.

\begin{figure}
\centering
  \includegraphics[width=\linewidth]{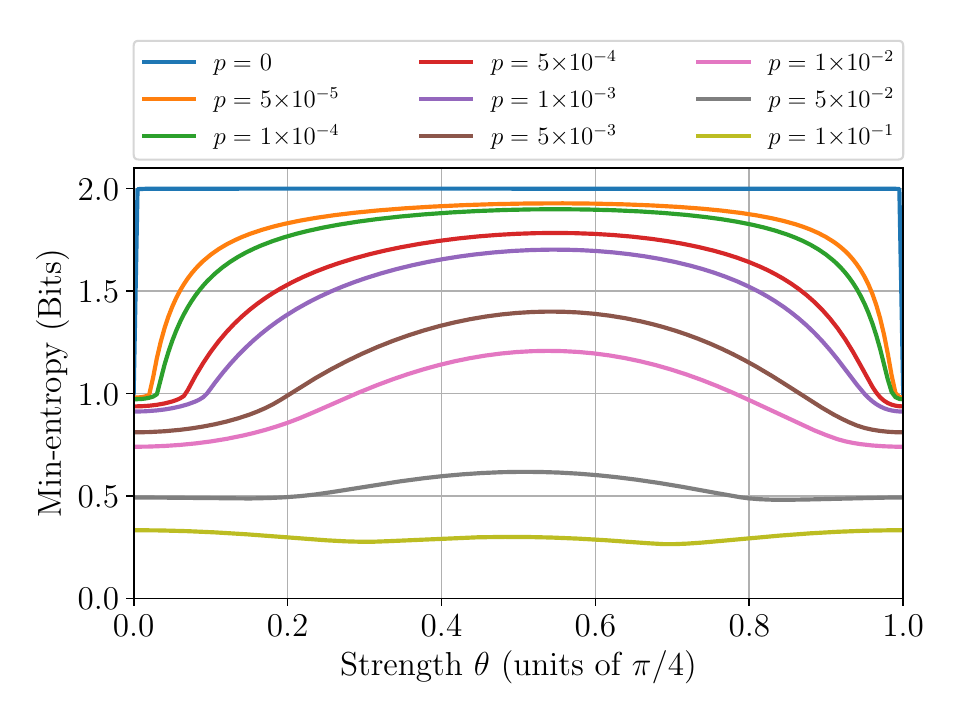}
  \caption{Min-entropy from the sequential protocol proposed in Sec. \ref{sec:seqCHSHprotocol} as a function of the strength $\theta$ for several values of the noise $p$. Simulations achieved with the NPA order $1+\text{AB}$.
  For $p=0$ we retrieve the results obtained
  analytically.}
  \label{fig:scanStrength}
\end{figure}
\begin{figure}
\centering
  \scalebox{0.53}{\input{figs/optimal_strength.pgf}}
  \caption{Best min-entropy achievable with three different protocols as a function of the depolarization parameter. Our proposal explained in Sec.\ \ref{sec:seqCHSHprotocol}, the standard CHSH protocol based on numerical optimization (NPA at order $4$) \cite{NietoSilleras2014}, and the sequential protocol proposed by \cite{Curchod2017a}. The experimental data are subjected to additional type of noise not considered by the curves simulations, such as the $c$ parameter discussed in Sec.\ \ref{sec:experiment}.}
  \label{fig:optimalStrength}
\end{figure}

In Fig.\ \ref{fig:scanStrength} we plot the simulation result, which confirms
that, in the ideal case ($p=0$), the min-entropy of the measurements of the protocol is two bits for each value of $\theta\in (0,\frac{\pi}{4})$. When the strength parameter $\theta$ is at one of the two extremes, the min-entropy drops to one bit, in agreement with our theoretical result.
With the help of the sequential protocol, it is straightforward to understand the drop by observing the state after the measurement of $\text{Bob}_1$.
For $\theta = 0$, $\text{Bob}_1$ measures projectively, hence the state sent to $\text{Bob}_2$ is separable and Eve can easily guess the second bit.
For $\theta = \frac{\pi}{4}$, the measurements of $\text{Bob}_1$ produce no useful correlations and their outcomes are also easily predictable by Eve.
Yet, because the measurement is non-interactive, $\text{Bob}_2$ still receives a portion of a maximally entangled pair and generates with Alice the perfect correlations that allow him to certify that his outcomes are unpredictable.
In both cases, one outcome (and hence one bit) is securely random, and the other is known to Eve.

Figure \ref{fig:scanStrength} also shows the impact that the noise quantified by $p$ has on the performance.
Intermediate values of $\theta$ are optimal, as they are farthest from the extremal points that reduce the randomness even in the ideal case.
The approximate flatness of the curve also means that inaccuracies in the setting of $\theta$ reduce performance only slightly, simplifying the requirements for the experimental implementation.
This descends from the fact that the performance of the noiseless protocol is independent of $\theta$ (except for the extremal points).
This is in contrast with all other protocols present in the literature, whose optimal performance is obtained for specific values of $\theta$ which are close to pathological points \cite{Curchod2017a, Foletto2021, Bowles2020}.

In Fig.\ \ref{fig:optimalStrength} we show the best min-entropy achievable with the sequential protocol as a function of the parameter $p$.
It indicates that it is possible to generate more than one random bit per state even if $p\approx 1.8 \cdot 10^{-2}$.
This value is fairly typical for sources of polarization-entangled photon pairs based on spontaneous parametric down-conversion, and can be reduced with state-of-the-art equipment \cite{Poh2015,Liu2018,Liu2021,Li2021,Liu2022}.
For comparison, we plot also the min-entropy achievable with a non-sequential protocol that works in the CHSH scenario and uses the NPA hierarchy \cite{NietoSilleras2014}.
We find that the threshold value of $p$ at which the two curves begin to split is approximately $8.5\cdot 10^{-2}$, meaning that for any smaller value the sequential protocol performs better than its non-sequential counterpart.
The equivalent threshold for the protocol of Ref.\ \cite{Curchod2017a} is a much smaller $3.7\cdot 10^{-3}$ \cite{Foletto2021}.

We point out that this value is in general affected by the finite orders of the NPA hierarchy set in the maximization \eqref{eq:SDP_problem} of the two protocols, which are $1+\text{AB}$ and $4$ respectively.

In a realistic implementation of the sequential scheme (which still neglects finite-size effects), Alice and the Bobs would generate random inputs to select the measurements to be performed on each state.
Their choices should be unbalanced, favoring $\mathbf{y}=\mathbf{y_r}=(0,1)$ for the Bobs, and arbitrarily one of the two observables for Alice.
This is to reduce the randomness cost to select the inputs, which, in the asymptotic limit can be made arbitrarily close to zero bits per state.
Alice and the Bobs' devices should receive the inputs and produce the outputs while outside of one another's light cones, to avoid the locality loophole.
From the complete list of inputs and outputs gathered in a time interval, Alice and the Bobs should calculate the experimental correlations $P_{\rm exp}(a,\mathbf{b}|x,\mathbf{y})$ to use in the SDP \eqref{eq:SDP_problem} with the help of the NPA hierarchy.
The string of outputs of the Bobs corresponding to $\mathbf{y}=\mathbf{y_r}$ should be considered as consisting of pairs of bits (one from Bob$_1$ and one from Bob$_2$).
The average min-entropy corresponding to each pair would be calculated from the guessing probability $G$ resulting from the problem.
Finally, the Bobs should reduce the string using a randomness extractor and the knowledge of the min-entropy, producing a shorter but uniform and secure sequence of random bits \cite{Trevisan2000}.
The post-processing, consisting of the SDP and the extraction, could be executed during the acquisition of further outputs for a subsequent experimental run, thus reducing its impact on performance.
However, the SDP for this protocol can typically be solved in seconds on average personal computers if tackled at level $1+\text{AB}$ of the NPA hierarchy.
This holds independently of the number of samples as only probabilities are used.
The extraction scales at worst quadratically with the length of the raw key, but can be efficiently parallelized \cite{Tang2019}.

\begin{table*}
    \input{table}
\end{table*}

\subsection{Experiment}
\label{sec:experiment}

We evaluated the protocol presented above with a proof-of-concept experiment, with the goal of verifying the feasibility of meeting the required quality for the entangled state and measurements.
For this purpose, we did not create an actual random number generator, but only a setup that reproduces all the quantum operations needed by the protocol, to observe the correlations.
Furthermore, we did not include the random inputs but only scanned all the measurement settings one by one.
Hence, our setup did not require any randomness source, which would be needed by a true generator.
As mentioned before, we can only infer probabilities from our experiment by assuming that the results for each quantum state are independent and identically distributed and neglecting the effects of a finite dataset.
We did not close neither the detection nor the locality loophole, relying instead on fair sampling and on the assumption that Alice and the Bobs do not communicate while producing outcomes (although Bob$_1$ is allowed to send information to Bob$_2$).
All of this can only be valid at the proof-of-concept level of our experiment and should be improved for a true implementation of the scheme. 
Yet, our observations are critical to show the feasibility and experimental robustness of the proposed protocol.

The experimental setup is the same of our previous works and uses polarization-entangled photon pairs and Mach-Zehnder interferometers to implement the Kraus operators \eqref{eq:bob1_kraus} \cite{Foletto2021,Foletto2020} (see also Supplementary Methods E for a detailed description).
Most of the imperfections in this setup can be modeled by a bipartite state of the form 
\begin{equation}
\label{eq:model_state}
    \rho_{AB}=(1-p-c)\dyad*{\phi^+}+p\frac{\openone}{4} + c\frac{\dyad{00}+\dyad{11}}{2} \,,
\end{equation}
where $p\in[0,1]$, as above, accounts for the depolarization caused by mixing with random noise, whereas $c\in[0,1]$ induces decoherence by reducing the extreme antidiagonal terms of the density matrix with respect to the diagonal ones.
In optical experiments, this is caused by alignment inaccuracies that increase the distinguishability between the two photons in each pair.
The two parameters $p$ and $c$ can be easily estimated experimentally by measuring the visibilities in the $\mathcal{Z}$ and $\mathcal{X}$ bases, indeed $p=1-V_{\mathcal{Z}}$ and $c=V_{\mathcal{Z}}-V_{\mathcal{X}}$ \cite{Foletto2021}.

We performed three experiments, labeled by an $\text{ID}\in\{1,2,3\}$.
Each of them attempts to reproduce the correlations required by the sequential-CHSH protocol described in Sec.\ \ref{sec:seqCHSHprotocol} and by the standard CHSH protocol.
For each experiment, we measured the correlations between Alice and the Bobs and we used them as constraints in an NPA hierarchy but instead of setting the whole statistic $P_{\rm exp}(a,\mathbf{b}|x,\mathbf{y})$, we constrained only the single-observable mean values $\expval*{A_x}$, $\expval*{B_{y_1}}$ and $\expval*{B_{y_1,y_2}}$, and the two-observable mean values $\expval*{A_xB_{y_1}}$ and $\expval*{A_xB_{y_1,y_2}}$, which are all obtainable from the experiment.
Doing so allowed us to get around the fact that our simplified experiment can produce results that do not strictly meet the requirements of the protocol.
Indeed, during the experiment the state produced by the source changes slightly.
This is mainly due to temperature variations that lead to the movement of the optical components.
This affects the interferometers and fiber couplings and eventually the experimental probability distribution.
Since we are scanning the measurements one by one, we are effectively using different states for each measurement, in contrast with Eq.\ \eqref{eq:prob_quant_seq}.
Constraining all correlations would have prevented the SDP from finding a proper solution, whereas our relaxed constraints allowed us to find one with a small solver tolerance of $10^{-12}$ \cite{Nakata2010}.
In general, this approach does not introduce security issues, since having a smaller number of constraints only gives more power to Eve and finds a min-entropy that is lower than what could be achieved by considering all the correlations.
The execution of the SDP was carried out on a personal computer and took less than \SI{10}{\second}.

We also compared the results with those predicted by our model using the same constraints, with the values of $p$, $c$, and $\theta$ that best fit the experimental data.
We calculated the statistical errors on the experimental results as standard deviations of a sample of $300$ simulated experiments.
In each of these, the photon counts descend from a Poisson distribution whose mean value is the experimental datum.

Tables \ref{tab:seqchsh_tests_results} and \ref{tab:chsh_tests_results} summarize the results of all three experiments, reporting the min-entropies and the mean values of the CHSH quantities $\expval{S_1}$, $\expval{S_2}$, $\expval{S_c}$, and $\expval{S}$ (which is measured in the non-sequential scenario).
They show that our protocol not only is feasible but can overcome the rate of the standard CHSH scheme in real world implementations.
Indeed, we found min-entropies between $0.82$ and $0.90$ bits, or between $23\%$ and $39\%$ higher than those obtained in the non-sequential scenario with the same states, even with visibilities $V_{\mathcal{Z}} \approx 98\%$ and $V_{\mathcal{X}}\approx 97\%$, which are readily accessible to entangled-photon sources built with commercial components.

In addition, the comparison between our results and the predictions of the model show that the latter can be used to evaluate the performance of this type of schemes.
The discrepancies can be attributed to other static imperfections in the setup which are not considered by the model and to the aforementioned changes of the state from one measurement to the next.

\section{Discussion}
In this work, we studied the set of sequential quantum correlations through a geometrical perspective.
Initially, we presented a general mathematical framework applicable to describing any sequential quantum scenario.
Using this framework in the context of one party on one side and two sequential parties on the other, we identified a  Tsirelson-like quantum boundary. 
This boundary can be interpreted as a  monogamy trade-off between the amounts of nonlocality of the sequential users shared with the remote one.
Despite this trade-off, we proved analytically that the correlations on the boundary certify the maximal amount of randomness in our device-independent scenario, specifically, two bits (excluding exceptional cases).
This result introduces a fundamental perspective: a trade-off for nonlocality does not necessarily translate into one for randomness.
In simpler terms, even if the correlations of one sequential user with the spatially separated one are explicable through local hidden variable theories, they can contribute to the generation of secure randomness when considered \textit{jointly} with the correlations of the other users.

We also proposed an explicit simple qubit-based protocol to generate the correlations on the boundary in the ideal case, and we numerically studied its noise robustness, finding that it can beat the non-sequential CHSH protocol for depolarization $p\lesssim 8.5\cdot 10^{-2}$ and produce more than one random bit for $p\lesssim 1.8\cdot10^{-2}$, values that are feasible to achieve with current technologies.

Finally, we implemented a proof-of-concept experiment, demonstrating not only the feasibility of our protocol, but also that it can perform better than the non-sequential CHSH-based scheme with real-world systems.
Indeed, we overcame the min-entropy of the latter by $23\%$ to $39\%$, and produced $0.90\pm 0.01$ bits in our best run. 
To the best of our knowledge, this is the first experimental observation of the advantage of a sequential protocol with respect to its one-step counterpart in terms of randomness generation.

On the basis of this work, we envisage further steps as follows.
When correlations lie on a quantum boundary, it may happen that they identify, or self-test, a unique (up to local isometries) quantum representation that realizes them \cite{Supic2020,Franz2011}.
It would be interesting to understand if this can happen also in the sequential case and whether the correlations of our protocol can self-test the state and measurements that produce them.

In addition, other portions of the boundary in this scenario might prove useful.
A possible avenue is to relax the condition $\expval*{S_c} = 2\sqrt{2}$ and study the bounds for $\expval*{S_\theta}$.
Our formalization of quantum sequential correlations in terms of commuting projective measurements might be helpful, but if boundary features cannot be analytically probed, the sequential extension of the NPA hierarchy can be used \cite{Bowles2020}.
It could also be meaningful to consider other parameterizations of the boundary.
For example, the upper bound of Eq.\ \eqref{eq:bellineqseq} can equivalently be written in terms of
\begin{equation}
    S'_\alpha \equiv \cos \alpha \ S_{+}+\sin \alpha \ S_{-}
\end{equation}
as
\begin{equation}
\label{eq:bellineqseq_alternative}
    \expval*{S'_\alpha} \leq 2\,,
\end{equation}
with 
$S_{\pm}=(\aliceobs{0} + \aliceobs{1})\bobobs{0} \pm (\aliceobs{0} - \aliceobs{1})\bobobs{0,1}$. \
\begin{figure}
    \centering
    \includegraphics[width=0.9\linewidth]{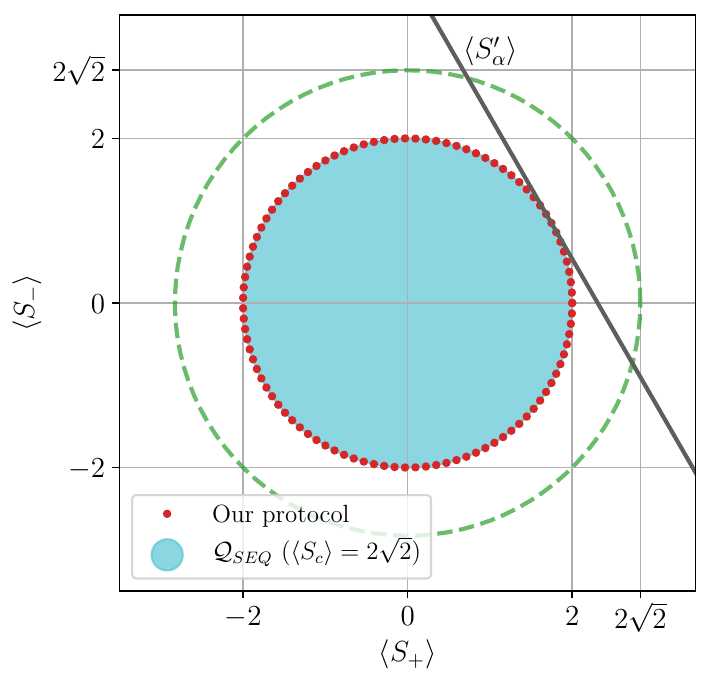}
    \caption{Cross-section of the sequential quantum set at $\expval{S_c}=2\sqrt{2}$ in the parametrization $\expval*{S_{\pm}}$.
    The dashed circumference denotes the maximum values achievable in non-sequential quantum scenarios, without restrictions on $\expval{S_c}$. The red dots mark the geometric location of the correlations achievable with the protocol explained in Sec \ref{sec:seqCHSHprotocol}.}
    \label{fig:boundary_S3}
\end{figure}
This expression, detailed in Supplementary Methods B, gives the boundary represented in Fig.\ \ref{fig:boundary_S3}.
Without constraining $\expval*{S_c} = 2\sqrt{2}$ and without the commutation relations of Eq.\ \eqref{eq:C_constraints} (derived from sequentiality), the Tsirelson-like bound of $\expval*{S_\alpha'}$ is relaxed to $2\sqrt{2}$, leading to a relation similar to the one in \cite{Christensen2015}.
Due to this greater similarity with the existing literature, $S'_\alpha$ might be easier to investigate than $S_\theta$.

Our protocol could also be more thoroughly investigated in its robustness to losses.
The standard way to treat losses in device-independent schemes is to assign the no-output events to one of the legitimate outputs.
In our case, this would cause the correlations to fall from the boundary and into the interior of $Q_{SEQ}$.
Could this be partially compensated for with a different set of states and measurements?

It would also be interesting to study whether the protocol can be extended to more Bobs.
This stems from the intuition that the independence of the min-entropy from the strength parameter is due to the sequence of two mutually-unbiased measurements, $\sigma_z$ and $\sigma_x$.
This opens up the possibility of adding a third sequential party measuring $\sigma_y$: In this case, the bits would be extracted from a sequence of three mutually unbiased observables.
Is it then possible to achieve three bits regardless of the strength parameters under ideal conditions?
Could the noise robustness of such a protocol be enough for real-world implementations?
A limitation might be the complexity of the SDP, which grows considerably with the number of Bobs.

In conclusion, this work offers tools and results that can improve our understanding of sequential quantum correlations and the performance of randomness generation protocols.
The formulation in terms of products of commuting measurements might provide a more intuitive description and suggest interesting points of view from which to analyze a given scenario.
For example, it can be used for the investigation of the sharing of nonlocality \cite{Cheng2021, Cheng2022, Steffinlongo2022}.
The boundary correlations we studied highlight that the greatest quantum advantage is reached using the entire set of experimental probabilities, and not just the pairwise ones.
This paves the way for further studies on the complex relationship between nonlocality and randomness and can improve the performance of device-independent random number generators with present-day technologies.

\section*{Data availability}
Data is available from the corresponding author upon reasonable request.

\section*{Code availability}
The codes used for the simulations for this paper are available from the corresponding author upon reasonable request.

\section*{Acknowledgements}
    The authors would like to thank Dr.\ Flavio Baccari (Max Planck Institute of Quantum Optics), Prof.\ Stefano Pironio (Université Libre de Bruxelles), and Dr.\ Peter Brown (Télécom Paris) for the useful discussions and clarifications.
    The computational resources offered by the CAPRI initiative (University of Padova Strategic Research Infrastructure Grant 2017: “CAPRI: Calcolo ad Alte Prestazioni per la Ricerca e l’Innovazione”) and BLADE cluster are acknowledged.
    Part of this work was supported by Ministero dell'Istruzione, dell'Università e della Ricerca (MIUR) (Italian Ministry of Education, University and Research) under the initiative ``Departments of Excellence'' (Law No. 232/2016), by Fondazione Cassa di Risparmio di Padova e Rovigo within the call ``Ricerca Scientifica di Eccellenza 2018'',  project QUASAR, by the European Union’s Horizon 2020 research and innovation programme, project QUANGO (grant agreement No. 101004341), and by  the project “Quantum Security Networks Partnership” (QSNP, grant agreement No. 101114043).

\section*{Author contributions}
M.P., G.F., L.C., and G.V. provided analytical proofs.
M.P. conducted numerical simulations.
M.P. carried out the experiment and analyzed the results, with assistance from G.F. and L.C. during setup preparation.
M.A., P.V., and G.V. supervised the work.
All authors participated in result discussions and contributed to the final manuscript.

\section*{Competing interests}
The authors declare no competing interests.

\bibliographystyle{naturemag}
\bibliography{bibliography}

\onecolumngrid
\appendix
\renewcommand{\thesubsection}{\Alph{subsection}}

\section*{Supplementary Methods}
\subsection{Alternative formulations of the sequential scenario}
\label{sec:app_seq_to_bipap}
Here we give an alternative characterization for sequential quantum correlations, which will serve to prove the validity of the contruction based on unitary and dichotomic observables introduced in the main text.

In the following we will use symbol $\mathbf{y}$ for a sequence of inputs and $\mathbf{y}_k$ for its truncation at the $k$-th element.
We will say that $\mathbf{y}_l \succeq \mathbf{y}_k$ if $l \geq k$ and the first $k$ elements in $\mathbf{y}_l$ are the same of $\mathbf{y}_k$ (i.e., $\mathbf{y}_k$ is a truncation of $\mathbf{y}_l$).

\begin{proposition}
    A given correlation $p(\mathbf{a}, \mathbf{b}| \mathbf{x},\mathbf{y})$ is sequential and quantum if and only if it can be written as $p(\mathbf{a}, \mathbf{b}| \mathbf{x},\mathbf{y}) = \expval{\prod_k \aliceproj{\mathbf{x}_k}{a_k} \otimes \prod_k \bobproj{\mathbf{y}_k}{b_k}}{\psi}$ with $\mathbf{x}\succeq \mathbf{x}_k, \mathbf{y}\succeq \mathbf{y}_k$, and the operators satisfying:
    \begin{equation}
        \begin{split}
            \sum_{b_k}\bobproj{\mathbf{y}_k}{b_k} &= \openone \quad \forall k, \mathbf{y}_k \quad \text{(normalization)} \\
            \bobproj{\mathbf{y}_k^\dagger}{b_k} &= \bobproj{\mathbf{y}_k}{b_k} \quad \forall k,  \mathbf{y}_k, b_k\quad \text{(hermiticity)} \\
            \bobproj{\mathbf{y}_k}{b_k}\bobproj{\mathbf{y}_k}{b'_k} &= \delta_{b_kb'_k} \bobproj{\mathbf{y}_k}{b_k} \quad \forall k, \mathbf{y}_k, b_k, b'_k \quad \text{(proj.\ and ortho.)} \\
            [\bobproj{\mathbf{y}_k}{b_k},\bobproj{\mathbf{y}_l}{b_l}] &= 0 \quad \forall k, l, b_k, b_l, \mathbf{y}_l\succeq \mathbf{y}_k \quad \text{(comm.)} \,.
        \end{split}
        \label{eq:proj_appendix}
    \end{equation}
    and similarly for $\aliceproj{\mathbf{x}_k}{a_k}$.
\end{proposition}

\begin{proof}
    Reference \cite{Bowles2020} already proves that any sequential quantum correlation can be written as $p(\mathbf{a}, \mathbf{b}| \mathbf{x},\mathbf{y}) = \expval{ \aliceseq{\mathbf{x}}{\mathbf{a}} \otimes \bobseq{\mathbf{y}}{\mathbf{b}}}{\psi}$ with the operators satisfying:
    \begin{equation}
        \begin{split}
            \sum_{\mathbf{b}}\bobseq{\mathbf{y}}{\mathbf{b}} &= \openone \quad \forall \mathbf{y}\quad \text{(normalization)} \\
            \bobseq{\mathbf{y}^\dagger}{\mathbf{b}} &= \bobseq{\mathbf{y}}{\mathbf{b}} \quad \forall \mathbf{y}, \mathbf{b}\quad \text{(hermiticity)} \\
            \bobseq{\mathbf{y}}{\mathbf{b}}\bobseq{\mathbf{y}}{\mathbf{b'}} &= \delta_{\mathbf{bb'}} \bobseq{\mathbf{y}}{\mathbf{b}} \quad \forall \mathbf{y}, \mathbf{b}, \mathbf{b'} \quad \text{(proj.\ and ortho.)} \\
            \sum_{b_{k+1}\ldots b_n}\bobseq{y_1\ldots y_{k+1}\ldots y_n}{\mathbf{b}} &= \sum_{b_{k+1}\ldots b_n}\bobseq{y_1\ldots y'_{k+1}\ldots y'_n}{\mathbf{b}} \\
            &\forall k, \mathbf{y}, \mathbf{y'}, b_1\ldots b_k \quad \text{(seq.)} \,.
        \end{split}
        \label{eq:seq_bowles_appendix}
    \end{equation}

    We can prove that any $\bobseq{\mathbf{y}}{\mathbf{b}}$ satisfying \eqref{eq:seq_bowles_appendix} can be written as the product $\prod_k \bobproj{\mathbf{y}_k}{b_k}$, with operators $\bobproj{\mathbf{y}_k}{b_k}$ satisfying \eqref{eq:proj_appendix} and viceversa to arrive at the conclusion.

    Let us start from $\bobseq{\mathbf{y}}{\mathbf{b}}$ satisfying \eqref{eq:seq_bowles_appendix}.
    Using all the truncations $\mathbf{y}_k$ of $\mathbf{y}$ (i.e., $\mathbf{y}\succeq \mathbf{y}_k$), we can define
    \begin{equation}
        \bobproj{\mathbf{y}_k}{b_k}\equiv \sum_{\mathbf{b'}}\delta_{b_k b'_k}\bobseq{\mathbf{y}}{\mathbf{b'}}\,.
        \label{eq:def_proj_from_seq}
    \end{equation}
    The product of these operators is $\prod_k \bobproj{\mathbf{y}_k}{b_k} = \sum_{\mathbf{b'}} \delta_{\mathbf{b}\mathbf{b'}} \prod_k \bobseq{\mathbf{y}}{\mathbf{b'}} = \bobseq{\mathbf{y}}{\mathbf{b}}$ as needed.
    
    Importantly, the right-hand side of \eqref{eq:def_proj_from_seq} is independent of $y_{k+1}\ldots y_n$.
    Indeed:
    \begin{equation}
        \sum_{\mathbf{b'}}\delta_{b_k b'_k}\bobseq{\mathbf{y}}{\mathbf{b'}} = \sum_{b'_1\ldots b'_{k-1}} \left( \sum_{b'_{k+1}\ldots b'_n} \bobseq{\mathbf{y}}{b'_1\ldots b'_{k-1}b_kb'_{k+1}\ldots b'_n }\right)
    \end{equation}
    and the term in parenthesis is exactly the one that appears in the sequentiality condition of \eqref{eq:seq_bowles_appendix} and is guaranteed to be independent of $y_{k+1}\ldots y_n$.
    Then, it is straightforward to verify the normalization, hermiticity, projectivity and orthogonality conditions in \eqref{eq:proj_appendix} starting from their counterparts in \eqref{eq:seq_bowles_appendix}.
    The commutation relation of \eqref{eq:proj_appendix} is found by noticing that 
    \begin{equation}
        \bobproj{\mathbf{y}_k}{b_k}\bobproj{\mathbf{y}_l}{b_l} = \bobproj{\mathbf{y}_l}{b_l}\bobproj{\mathbf{y}_k}{b_k} = \sum_{\mathbf{b'}} \delta_{b_kb'_k}\delta_{b_lb'_l}\bobseq{\mathbf{y}}{\mathbf{b'}}\,,
    \end{equation}
    which uses the fact that $\bobseq{\mathbf{y}}{\mathbf{b}}$ are orthogonal projectors and $\mathbf{y} \succeq \mathbf{y}_l\succeq \mathbf{y_k}$.

    Viceversa, let us start from projectors $\bobproj{\mathbf{y}_k}{b_k}$ satisfying \eqref{eq:proj_appendix}.
    For any given pair of input and output sequences $\mathbf{y}, \mathbf{b}$, we can directly define:
    \begin{equation}
        \bobseq{\mathbf{y}}{\mathbf{b}} \equiv \prod_l \bobproj{\mathbf{y}_l}{b_l}\,.
    \end{equation}
    Normalization, projectivity and orthogonality are straightforward.
    Hermiticity descends from the fact that the product involves only commuting hermitian operators.
    The sequentiality condition is verified considering that
    \begin{equation}
        \sum_{b_{k+1}\ldots b_n} \bobseq{\mathbf{y}}{\mathbf{b}} = \prod_{l=1}^k \bobproj{\mathbf{y}_l}{b_l} \prod_{l=k+1}^n \left(\sum_{b_l} \bobproj{\mathbf{y}_l}{b_l} \right) = \prod_{l=1}^k \bobproj{\mathbf{y_l}}{b_l}
    \end{equation}
    which does not depend on $y_{k+1}\ldots y_n$.
\end{proof}

All of this is valid regardless of the number of possible values for the inputs or outputs.
In the special case in which all measurements are dichotomic and return $b_k \in \{\pm 1\}$, we can build observables as
\begin{equation}
    \bobobs{\mathbf{y}_k}\equiv \sum_{b_k} b_k \bobproj{\mathbf{y}_k}{b_k} = \sum_{\mathbf{b}}b_k\bobseq{\mathbf{y}}{\mathbf{b}}
\end{equation}
which are hermitian and unitary.
This means that the construction used in the main text is a valid way to characterize sequential quantum correlations (in the scenario of interest for this work).
For completeness, we report that, in this dichotomic case, the operators $\bobseq{\mathbf{y}}{\mathbf{b}}$ can conversely be built from the observables with:

\begin{equation}
    \bobseq{\mathbf{y}}{\mathbf{b}} = \frac1{2^n}\Bigl(\openone+\sum_{k_1}b_{k_1}\bobobs{\mathbf{y}_{k_1}}+\sum_{k_1<\cdots <k_n}b_{k_1}\cdots b_{k_n}    \prod_{k = k_1}^{k_n}\bobobs{\mathbf{y}_{k}}\Bigr)\,.
\end{equation}

To summarize, the relevant operators in our scenario satisfy:
\begin{equation}
    \begin{gathered}
        \aliceobs{x}^\dagger \aliceobs{x} = \openone \quad \forall x \quad \text{(unit.)}\\
        \aliceobs{x}  = \aliceobs{x}^\dagger \quad \forall x \quad \text{(herm.)} \\
        \bobobs{y_1}^\dagger \bobobs{y_1} = \bobobs{y_1,y_2}^\dagger \bobobs{y_1,y_2} = \openone  \quad \forall y_1,y_2 \quad\text{(unit.)} \\
        \bobobs{y_1} = \bobobs{y_1}^\dagger,\, \bobobs{y_1,y_2} = \bobobs{y_1,y_2}^\dagger \forall y_1,y_2 \quad \text{(herm.)}\\
        [\bobobs{y_1},\bobobs{y_1,y_2}] = 0 \quad \forall y_1, y_2 \quad \text{(commutation)}\,,        
    \end{gathered}
    \label{eq:our_cons_appendix}
\end{equation}
where it is left implicit that $\aliceobs{x}$ and $\bobobs{y_1},\bobobs{y_1,y_2}$ act on separate Hilbert spaces.

\subsection{Characterization of the sequential quantum boundary}
\label{sec:app_charac_boundary}
Here we prove Result 1, i.e. that the relations $\expval*{S_c} = 2\sqrt{2}$ and $\expval*{S_\theta} = \sqrt{2}$ are a boundary for $Q_{SEQ}$. We also give an useful characterization for the states that allow to generate correlations on this boundary.
We rephrase the result in a more self-contained way as:
\begin{proposition}
    Let $x, y_1, y_2 \in \{0,1\}$ and $\aliceobs{x}$, $\bobobs{y_1}$, $\bobobs{y_1,y_2}$ be operators that satisfy \eqref{eq:our_cons_appendix}.
    Define the operators (as in the main text)
    \begin{equation}
        \begin{split}
            S_1&\equiv(\aliceobs{0} + \aliceobs{1})\bobobs{0} +(\aliceobs{0} - \aliceobs{1})\bobobs{1}\\
            S_2&\equiv(\aliceobs{0} + \aliceobs{1})\bobobs{0,0} +(\aliceobs{0} - \aliceobs{1})\bobobs{0,1} \\
            S_c&\equiv(\aliceobs{0} + \aliceobs{1})\bobobs{0,0} +(\aliceobs{0} - \aliceobs{1})\bobobs{1} \\
            S_\theta &\equiv \cos2\theta (S_1-\sqrt{2}\openone) + \sin 2\theta (S_2-\sqrt{2}\openone)\, ,
        \end{split}
    \end{equation}
    Then the following two relations hold
    \begin{equation}
        \begin{gathered}
            \expval*{S_c} \leq 2\sqrt{2} \\
            \expval*{S_c} = 2\sqrt{2} \implies \expval*{S_\theta} \leq \sqrt{2}
        \end{gathered}
        \label{eq:boundary_appendix}
    \end{equation}
    and the inequalities are tight for any $\theta$.    
\end{proposition}
This means that $\expval*{S_c} = 2\sqrt{2}$ and $\expval*{S_\theta} = \sqrt{2}$ define a boundary for the set of sequential quantum correlations $Q_{SEQ}$ measurable in the scenario of interest of this work.

\begin{proof}
    Since $S_c$ is a CHSH-like operator, Tsirelson's bound already assures that $\expval*{S_c}\leq 2\sqrt{2}$.
    Therefore, we move to proving that if $\expval*{S_c} = 2\sqrt{2}$, then $\expval*{S_\theta} \leq \sqrt{2}$.
    Then, let $\ket{\psi}$ and the $\aliceobs{x}$, $\bobobs{y_1}$, $\bobobs{y_1,y_2}$ be a state and observables for which $\expval*{S_c} = 2\sqrt{2}$.
    Let us define for ease of writing the operators $Z_A = \frac{\aliceobs{0}+\aliceobs{1}}{\sqrt{2}}$ and $X_A = \frac{\aliceobs{0}-\aliceobs{1}}{\sqrt{2}}$.
    By construction we have that $\{Z_A,X_A\}=0$.
    Moreover, due to the self-testing properties of the CHSH scenario provided by $\expval**{S_c} = 2\sqrt{2}$, we have that $\expval*{Z_A \bobobs{0,0}} = \expval*{X_A \bobobs{1}} = 1$ and $\{\aliceobs{0},\aliceobs{1}\}\ket{\psi}=0$ \cite{Supic2020}, which implies $Z_A^{\dagger}Z_A\ket{\psi} = X_A^{\dagger}X_A\ket{\psi} = \ket{\psi}$.

    With these properties, we can rewrite the mean value of $S_\theta$ on $\ket{\psi}$ as
    \begin{equation}
    \begin{split}
        \expval*{S_\theta} &=\sqrt{2} \cos2\theta\expval*{Z_A \bobobs{0}}+\sqrt{2}\sin2\theta\expval*{X_A \bobobs{0,1}}\,.
    \end{split}
    \end{equation}

    Furthermore, let us define the auxiliary hermitian operator
    \begin{equation}
        P \equiv 2^{-\frac14} \left(\openone -\cos(2\theta) Z_A\bobobs{0} -\sin(2\theta)X_A\bobobs{0,1}\right)
        \label{eq:p_definition_appendix} \ .
    \end{equation}
    With an algebraic derivation we find that
    \begin{equation}
        \expval*{P^2} = \expval*{\sqrt{2}\openone - S_\theta} \ .
    \end{equation}

    Since $P^2$ is the square of an hermitian operator, it has non-negative eigenvalues.
    Hence $\expval*{P^2}\geq 0 $ and $\expval*{S_\theta} \leq \sqrt{2}$.

    We can also find an explicit quantum protocol that satisfies $\expval*{S_c} = 2\sqrt{2}$ and $\expval*{S_\theta} = \sqrt{2}$ for any $\theta$, concluding that the inequalities are tight.
    Because this protocol is of relevance for different parts of this text, we show it separately in Section \ref{sec:app_dilatation}.    
\end{proof}

We can also prove the following characterization for $\ket{\psi}$:
\begin{proposition}
    Let $\ket{\psi}$, $\aliceobs{x}$, $\bobobs{y_1}$,$\bobobs{y_1,y_2}$ be a state and operators (with the constraints \eqref{eq:our_cons_appendix}) that produce correlations such that $\expval*{S_c} = 2\sqrt{2}$ and $\expval*{S_\theta} = \sqrt{2}$. Then
    \begin{equation}
    \label{eq:state_satur_Stheta}
        \ket{\psi} = \cos2\theta \frac{\aliceobs{0}+\aliceobs{1}}{\sqrt{2}} \bobobs{0}\ket{\psi} + \sin2\theta\frac{\aliceobs{0}-\aliceobs{1}}{\sqrt{2}} \bobobs{0,1}\ket{\psi} \,,
    \end{equation}
\end{proposition}
\begin{proof}
    Using the auxiliary hermitian operator $P$ of Eq.\ \eqref{eq:p_definition_appendix} and considering that $\ket{\psi}$ generates correlations on the boundary, we find:
    \begin{equation}
        \expval*{\sqrt{2}\openone - S_\theta} = \expval*{P^2} = \expval*{P^\dagger P} = \norm{P\ket{\psi}}^2 = 0
    \end{equation}
    which implies $P\ket{\psi}= 0$ and hence \eqref{eq:state_satur_Stheta}.
\end{proof}

In these proofs we used the self-testing properties of the CHSH inequality to fix some expectation values. 
From the same observation, we can obtain alternative and equivalent formulations of the boundary. 
For example, when $\expval*{S_c}=2\sqrt{2}$, we can introduce the operator 
\begin{equation}\label{eq:Sprime_alpha}
    S'_\alpha=\cos \alpha \ S_{+}+\sin \alpha \ S_{-}
\end{equation}
with 
\begin{equation}
    S_{\pm}=(\aliceobs{0} + \aliceobs{1})\bobobs{0} \pm (\aliceobs{0} - \aliceobs{1})\bobobs{0,1} \ .
\end{equation}
Under the condition $\expval*{S_c}=2\sqrt{2}$, we have that
\begin{equation}
    \sqrt{2}\expval{S_\theta}=\expval*{S'_{\frac{\pi}{8}-\theta}}
\end{equation}
so that the sequential set can be also characterized by the condition  $\expval{S'_\alpha}\le 2$.

Expression \eqref{eq:Sprime_alpha} is similar to those studied in \cite{Christensen2015} in a non-sequential scenario, and indeed has an analogous sum-of-squares decomposition
\begin{equation}\label{eq:SOS_alternative}
\begin{split}
    2\sqrt{2}\openone-S'_\alpha =& 
    \frac{1}{\sqrt{2}}\left[\sin\left(\frac{\pi}{4}+\alpha \right)B_0+\cos\left(\frac{\pi}{4}+\alpha \right)B_{0,1}-A_0\right]^2+\\
    &+\frac{1}{\sqrt{2}}\left[\sin\left(\frac{\pi}{4}+\alpha \right)B_0-\cos\left(\frac{\pi}{4}+\alpha \right)B_{0,1}-A_1\right]^2 \ .
\end{split}
\end{equation}
Having a sum of non-negative operators on the right hand side, Eq.\ \eqref{eq:SOS_alternative} gives a bound for the maximum value allowed by quantum physics for the expectation value of $ S_\alpha '$, independently on the value of $S_c$:
\begin{equation}\label{eq:bound_Sprime}
    \expval{S_\alpha'}\leq 2\sqrt{2} \ .
\end{equation}
This bound is actually tight since we can choose the strategy

\begin{equation}
    \begin{split}
        &A_j=(-1)^j\cos(\alpha+\frac{\pi}{4}) \sigma_x+\sin(\alpha+\frac{\pi}{4}) \sigma_z \\
        &B_0=\sigma_z \ , \qquad B_{0,1}=\sigma_x 
    \end{split}
\end{equation}
with the shared entangled state $\ket{\phi^+}$ to saturate it. 
The condition \eqref{eq:bound_Sprime} is represented with a dashed line in Fig.\ 4 of the main text. 

Different parameterizations of the boundary and arguments like the one just shown are not necessary for the proofs of this paper but could be useful when the condition $\expval*{S_c}=2\sqrt{2}$ is relaxed and the interior of the sequential set is studied.

\subsection{Proof of the randomness results}
\label{sec:app_minentropy}
Here we prove Result 2 about the randomness of the outcomes of $\bobobs{0}$ and $\bobobs{0,1}$.
We rephrase it more formally as:
\begin{proposition}
    Let $x, y_1, y_2 \in \{0,1\}$, $a, b_1, b_2 \in \{\pm 1\}$ and let $P_{exp}(a,\mathbf{b}|x,\mathbf{y}) = \expval*{\aliceproj{x}{a}\otimes \bobproj{y_1}{b_1}\bobproj{y_1,y_2}{b_2}}{\psi}\in Q_{SEQ}$, where $\aliceproj{x}{a}$, $\bobproj{y_1}{b_1}$ and $\bobproj{y_1,y_2}{b_2}$ are the projectors on the eigenspaces of observables $\aliceobs{x}$, $\bobobs{y_1}$ and $\bobobs{y_1,y_2}$ which satisfy \eqref{eq:our_cons_appendix}.
    Let $G$ be the guessing probability defined as
    \begin{gather}
    \label{eq:SDP_problem}
        G = \max_{p_{ABE}} \sum_{\mathbf{b}} p_{BE} \qty( \mathbf{b},\mathbf{b}|\mathbf{y_r} ) \\
    \begin{aligned}
    \label{eq:SDP_constraints}
        \text{s.t.} \hspace{0.5cm} \sum_{\mathbf{e}} p_{ABE}(a,\mathbf{b},\mathbf{e}|x,\mathbf{y}) &= P_{\rm exp}(a,\mathbf{b}|x,\mathbf{y})\, , \\
        p_{ABE}(a,\mathbf{b},\mathbf{e}|x,\mathbf{y}) &\in \mathcal{Q}_{SEQ}\,.
    \end{aligned}
    \end{gather}
    for $\mathbf{y_r} = (0,1)$.

    Then, if $\expval*{S_c}= 2\sqrt{2}$ and $\expval*{S_\theta} = \sqrt{2}$, we have that $G = \frac14$ if $\theta \neq n \frac{\pi}{4}$ and $G= \frac12$ otherwise.
\end{proposition}
\begin{proof}
    We work in the formalism of device-independent scenarios, which assumes that in principle the state shared by Alice and the Bobs is not separable from Eve's. We denote it $\ket{\psi}_{ABE} \in \mathcal{H}_A\otimes\mathcal{H}_B\otimes\mathcal{H}_E$.
    Alice's observables $\aliceobs{0}$ and $\aliceobs{1}$ act on $\mathcal{H}_A$, the Bobs' observables $\bobobs{0}$, $\bobobs{1}$, $\bobobs{0,0}$ and $\bobobs{0,1}$ act on $\mathcal{H}_B$.
    Since $\bobobs{1,0}$ and $\bobobs{1,1}$ contribute neither to $S_c$ nor to $S_\theta$, they have no role in this proof.
    Eve performs projective measurements $\eveproj{\mathbf{e}}$ on $\mathcal{H}_E$, that satisfy $\eveproj{\mathbf{e}}\eveproj{\mathbf{e'}}=\delta_{\mathbf{e}\mathbf{e'}}\eveproj{\mathbf{e}}$ and $\sum_{\mathbf{e}} \eveproj{\mathbf{e}} = \openone_E$.
    The joint correlation including Alice, Bob and Eve is 
    \begin{equation}
    p(a,\mathbf{b},\mathbf{e}|x,\mathbf{y}) = \expval*{\aliceproj{x}{a}\otimes \bobproj{y_1}{b_1}\bobproj{y_1,y_2}{b_2}\otimes \eveproj{\mathbf{e}}}{\psi} \, .
    \end{equation}
    
    Following \cite{Supic2020}, the saturation $\expval*{S_c}=2\sqrt{2}$ implies that, up to local isometries, the state can be written as $\ket{\psi}_{ABE} \uli \ket*{\phi^+}_{AB'}\ket{\xi}_{B''E} \in \mathcal{H}_A\otimes\mathcal{H}_{B'}\otimes\mathcal{H}_{B''} \otimes\mathcal{H}_E$, where the symbol $\uli$ is introduced to distinguish between equality and the equivalence up to local isometries used in the self-testing formalism.
    These local isometries do not act on Eve's space.
    With the same formalism we have that Alice's operators are such that
    \begin{equation}
        \begin{split}
            \aliceobs{0} \ket{\psi} &\uli \frac{\sigma_z + \sigma_x}{\sqrt{2}} \ket*{\phi^+}\ket{\xi} \\ 
            \aliceobs{1} \ket{\psi} &\uli \frac{\sigma_z - \sigma_x}{\sqrt{2}} \ket*{\phi^+}\ket{\xi} \ .
        \end{split}
    \end{equation}
    Since $\mathcal{H}_{B'}$ is a qubit space, $\mathcal{H}_{B'}=\mathbb{C}^2$, the action of the operators $\bobobs{0}$ and $\bobobs{0,1}$ can be decomposed with the Pauli matrices as
    \begin{equation}
        \begin{split}
            \bobobs{0}\ket{\psi} &\uli \Bigl[\openone\otimes \gamma_0 + \sigma_x\otimes\gamma_1 + \sigma_y\otimes\gamma_2 + \sigma_z\otimes\gamma_3\Bigr]\ket*{\phi^+}\ket{\xi}  \\
            \bobobs{0,1}\ket{\psi} &\uli \Bigl[\openone\otimes \tau_0 + \sigma_x\otimes\tau_1 + \sigma_y\otimes\tau_2 + \sigma_z\otimes\tau_3\Bigr]\ket*{\phi^+}\ket{\xi} \,,
        \end{split}
    \end{equation}
    where the $\gamma_i$ and $\tau_i$ operators are arbitrary hermitian operators acting on $\mathcal{H}_{B''}$.
    Since $\bobobs{0,0}\ket{\psi}\uli\sigma_z\otimes\openone\ket*{\phi^+}\ket{\xi}$ (because $\expval*{S_c} = 2\sqrt{2}$), the commutation relation $[\bobobs{0},\bobobs{0,0}]=0$ implies that $\gamma_1\ket{\xi} = \gamma_2\ket{\xi} = 0$.
    
    Then, by applying the local isometry that realizes our relations to both sides of  \eqref{eq:state_satur_Stheta}, we retrieve the following four constraints 
    
    \begin{align}
       \label{eq:gamma_tau_first} \cos2\theta \gamma_0 \ket{\xi} - i\sin2\theta \tau_2\ket{\xi} &= 0 \\
        \sin2\theta \tau_0\ket{\xi} &= 0 \\
        \sin2\theta \tau_3\ket{\xi} &= 0 \\
        \cos2\theta \gamma_3 \ket{\xi} + \sin2\theta \tau_1\ket{\xi} &= \ket{\xi}\,,
    \end{align}
    and consequently, if $\sin 2\theta \neq 0$, $\tau_0\ket{\xi}=\tau_3\ket{\xi}=0$, $\bobobs{0,1}\ket{\psi} \uli \Bigl[\sigma_x \otimes \tau_1 + \sigma_y \otimes \tau_2\Bigr]\ket*{\phi^+}\ket{\xi}$.
    
  Let us now study the guessing probability:
    \begin{equation}
        \begin{split}
            G &= \sum_{b_1b_2} \text{Prob}[e_1=b_1, e_2=b_2|y_1=0,y_2=1] \\
            &= \sum_{b_1b_2} \expval*{\openone_A\otimes \bobproj{0}{b_1} \bobproj{0,1}{b_2}\otimes \eveproj{b_1,b_2}}{\psi} \ .
        \end{split}
    \end{equation}
    
    We consider first the case $\sin 2\theta\neq 0$ and $\cos 2\theta\neq 0$.
    We can write the projectors that decompose $\bobobs{0}$ and $\bobobs{0,1}$ as:
    
    \begin{align}
        \bobproj{0}{b_1}\ket{\psi} &\uli \frac12\Bigl[\openone\otimes\openone + {b_1}\left(\openone\otimes\gamma_0 +\sigma_z\otimes\gamma_3\right)\Bigr]\ket*{\phi^+}\ket{\xi} \\ 
        \bobproj{0,1}{b_2}\ket{\psi} &\uli \frac12\Bigl[\openone\otimes\openone + {b_2}\left(\sigma_x\otimes\tau_1 +\sigma_y\otimes\tau_2\right)\Bigr]\ket*{\phi^+}\ket{\xi}\label{eq:c2b2simplified}
    \end{align}
    and replace them, together with the form of $\ket{\psi}$ in $G$.
    Since all terms in which $\bobproj{0}{b_1}\bobproj{0,1}{b_2}$ acts as a Pauli matrix on $\mathcal{H}_{B'}$ vanish because $\expval*{\openone_A\otimes\sigma_i}{\phi^+}=0$, we end up with:
    \begin{equation}
        G=\frac14\left(1+\sum_{b_1b_2}{b_1}\expval*{\gamma_0\otimes \eveproj{b_1,b_2}}{\xi}\right) \ .
    \end{equation}
    
    Then, we multiply from the left both sides of Eq.\ \eqref{eq:gamma_tau_first} by $\bra{\xi}\eveproj{b_1,b_2}$, finding that 
    \begin{equation}
        \cos2\theta\expval*{\eveproj{b_1,b_2}\otimes\gamma_0}{\xi} = i \sin2\theta\expval*{\eveproj{b_1,b_2}\otimes\tau_2}{\xi}\,.
    \end{equation}
    Considering that $\gamma_0$ and $\tau_2$ are hermitian and commute with Eve's operators (because they act on different Hilbert spaces), the expectation values are real.
    Therefore, taking the real part on both sides, we have that 
    \begin{equation}
        \cos2\theta\expval*{\eveproj{b_1,b_2}\otimes\gamma_0}{\xi} = \cos2\theta\expval*{\gamma_0\otimes \eveproj{b_1,b_2}}{\xi} = 0\,.
    \end{equation}
    If $\cos 2\theta \neq 0$, we have that
    \begin{equation}
    \label{eq:gamma0Eis0}
        \expval*{\gamma_0\otimes \eveproj{b_1,b_2}}{\xi} = 0
    \end{equation}
    and hence $G=\frac14$.
    
    For $\sin 2\theta =0$, we consider the probability that Eve guesses the outcome of $\bobobs{0}$ (regardless of whether she guesses $\bobobs{0,1}$ or not):
    \begin{equation}
        \begin{split}
            G_1 &= \sum_{b_1} \text{Prob}[e_1=b_1|y_1=0,y_2=1] \\
            &= \sum_{b_1} \expval*{\openone_A\otimes \bobproj{0}{b_1} \otimes \sum_{b_2}\eveproj{b_1,b_2}}{\psi}
        \end{split}
    \end{equation}
    We use the same decomposition of $\bobproj{0}{b_1}$ of above.
    The terms in $\openone\otimes\openone$ and $\openone\otimes\gamma_0$ are the only ones that do not vanish, and we have $G_1 = \frac12\left(1+\sum_{b_1b_2} {b_1} \expval*{\gamma_0\otimes \eveproj{b_1,b_2}}{\xi}\right)$.
    Because $\cos 2\theta \neq 0$, we can use Eq.\ \eqref{eq:gamma0Eis0} to find that the second term in the above parenthesis is 0 and $G_1=\frac12$.
    Since by definition of joint probability $G\leq G_1$, we have that $G\leq \frac12$.
    In Section \ref{sec:app_dilatation}, we show a strategy with which Eve can perfectly guess the outcome of $\bobobs{0,1}$, hence $G=\frac12$.
    
    For $\cos 2\theta = 0$, we consider the probability that Eve guesses the outcome of $\bobobs{0,1}$ (regardless of whether she guesses $\bobobs{0}$ or not):
    \begin{equation}
        \begin{split}
            G_2 &= \sum_{b_2} \text{Prob}[e_2=b_2|y_1=0,y_2=1] \\
            &= \sum_{b_2} \expval*{\openone_A\otimes \bobproj{0,1}{b_2}\otimes \sum_{b_1}\eveproj{b_1,b_2}}{\psi} \ .
        \end{split}
    \end{equation}
    Because $\sin 2\theta \neq 0$, we can replace Eq.\ \eqref{eq:c2b2simplified} in $G_2$.
    Then all the terms acting as a Pauli matrix on $\mathcal{H}_{B'}$ vanish and $G_2 = \frac12\sum_{b_1b_2}\expval*{\openone_A\otimes\openone_{B'}\otimes\openone_{B''}\otimes \eveproj{b_1,b_2}}{\psi}=\frac12$.
    Since by definition of joint probability $G\leq G_2$, we have that $G\leq \frac12$.
    In Section \ref{sec:app_dilatation}, we show a strategy with which Eve can perfectly guess the outcome of $\bobobs{0}$, hence $G=\frac12$.
\end{proof}

\subsection{Sequential-CHSH protocol with projective measurements}
\label{sec:app_dilatation}

In this section we discuss in more detail the protocol presented in the main text, justifying the equivalence between the formulations in terms of Kraus operators and of observables $\bobobs{y_1}$,$\bobobs{y_1,y_2}$. 
The natural framework to describe such equivalence is the Von Neumann formalism, in which the measurements are realized through the interaction of the system with the environment.
To simplify the notation, we focus only on the Bobs' side, where we realize the sequential measurements.

A reference scheme is depicted in Fig. \ref{fig:seq_setup}.
We denote the state of the system with $\ket{\varphi} \in \mathcal{H}_{B'}$ and we couple it with an ancilla qubit state  $\ket{\theta}\in \mathcal{H}_{B''}$,
\begin{equation}    
    \ket{\theta} \equiv \cos\theta \ket{0} + \sin\theta \ket{1} \,,
\end{equation}
forming the joint state $\ket{\psi}\equiv\ket{\varphi}\otimes\ket{\theta}=\ket{\varphi}\ket{\theta}$. 

When the input of Bob$_1$ is $y_1=1$, he's performing the projective measurement of $\bobobs{1}=\sigma_x \otimes \openone$ on $\mathcal{H}_{B'} \otimes \mathcal{H}_{B''}$.
When, instead, $y_1=0$, the measurements can be described through the unitary evolution
\begin{equation}
\begin{split}
    U\ket{\psi}&=e^{i \mathcal{H}_I}\ket{\varphi}\ket{\theta} =  e^{i \frac{\pi}{4}\qty( \openone-\sigma_z)\otimes\qty(\openone-\sigma_x ) }\ket{\varphi}\ket{\theta} = \\ &=\Bigl[\dyad{0}{0}\otimes \openone+\dyad{1}{1}\otimes \sigma_x\Bigr]\ket{\varphi}\ket{\theta}\, .
\end{split}
\end{equation}
Here the interaction Hamiltonian $\mathcal{H}_I$ defines a CNOT gate between the original state and the ancilla.
The Kraus operators realizing Bob$_1$'s measurements can be found
by measuring in the ancilla basis defined by the states $\ket{0}$ and $\ket{1}$, which here are the eigenstates of $\sigma_z$ in $\mathcal{H}_{B''}$:
\begin{equation}
    \begin{split} 
        p(b_1=+1 \lvert y_1=0) &= \norm{{\mel{0}{U}{\theta}}\ket{\varphi}}^2 = \norm{K_{+}(\theta)\ket{\varphi}}^2\\
        p(b_1=-1\lvert y_1=0) &= \norm{{\mel{1}{U}{\theta}}\ket{\varphi}}^2  = \norm{K_{-}(\theta)\ket{\varphi}}^2
    \end{split}
    \label{eq:app_prob_bob_1}
\end{equation}
where $K_{+} (\theta) \equiv {\mel{0}{U}{\theta}}$ and $K_{-} (\theta) \equiv {\mel{1}{U}{\theta}}$ are exactly the operators in defined in the main text, acting on the qubit space $\mathcal{H}_{B'}$. 
These Kraus operators are chosen so that for $\theta=\frac{\pi}{4}$ are proportional to the identity operator while for $\theta=0$ they realize a projective measurement of $\sigma_z$ (assigning the labels $\pm 1$ to the two outcomes).
In the literature, these Kraus operators are associated to weak measurements of $\sigma_z$ \cite{Foletto2020}.
Indeed, the unitary evolution can be also written as $U\ket{\psi}=U_\theta\ket{\varphi}\ket{0}$ with $U_\theta=e^{i \mathcal{H}_I}(\openone\otimes e^{-i\frac{\theta}2\sigma_y})$ and the $\theta$ parameter representing the ``strength'' of the interaction.

To obtain the projective description in terms of observables $\bobobs{y_1}$ and $\bobobs{y_1,y_2}$, we can recast \eqref{eq:app_prob_bob_1} as
\begin{equation}
    \begin{split}
        p(b_1=+1 \lvert y_1=0)&= \bra{\psi}U^\dagger \Bigl[\openone \otimes \dyad{0}\Bigr]U \ket{\psi} \\
        p(b_1=-1 \lvert y_1=0)&= \bra{\psi}U^\dagger \Bigl[\openone \otimes \dyad{1}\Bigr]U \ket{\psi}
    \end{split}
\end{equation}
and identify
\begin{equation}
    \begin{split}
        \bobproj{0}{+} &\equiv U^\dagger \Bigl[\openone \otimes \dyad{0}\Bigr]U \\
       \bobproj{0}{-} &\equiv U^\dagger \Bigl[\openone \otimes \ketbra{1}{1}\Bigr]U \ .
    \end{split}
\end{equation}
Substituting the explicit form of $U$, we find that $\bobproj{0}{+}$ and $\bobproj{0}{-}$ are the projectors associated with the eigenvalues $\pm 1$ of the observable $\bobobs{0}=\bobproj{0}{+}-\bobproj{0}{-}=\sigma_z\otimes \sigma_z$. 

\begin{figure}
    \centering
    \includegraphics[width=0.4\linewidth]{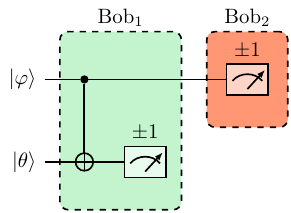}
    \caption{The sequential framework can be described adopting the Von Neumann formalism, where the system is coupled with an ancillary state through an unitary evolution, which is a CNOT gate in this case.}
    \label{fig:seq_setup}
\end{figure}

As a second step, Bob$_2$ receive the state after the unitary evolution. 
In correspondence of the input $y_2=0$ he performs projective measurements of $\sigma_z$, producing the statistics
\begin{equation}
    \begin{split}
        p(b_2=+1 \lvert \mathbf{y}=0,0) &= \expval*{U^{\dagger}\qty[\frac{\left(\openone+\sigma_z \right)}{2}\otimes \openone ] U}{\psi}  \\
        p(b_2=-1 \lvert \mathbf{y}=0,0) &= \expval*{U^{\dagger}\qty[\frac{\left(\openone-\sigma_z \right)}{2}\otimes \openone ] U}{\psi}
    \end{split}
\end{equation}
The definitions
\begin{equation}
    \begin{split}
         \bobproj{0,0}{+} &\equiv U^{\dagger}\qty[\frac{\left(\openone+\sigma_z \right)}{2}\otimes \openone ] U \\
        \bobproj{0,0}{-} &\equiv U^{\dagger}\qty[\frac{\left(\openone-\sigma_z \right)}{2}\otimes \openone ] U
    \end{split}
\end{equation}
lead to the projectors of $\bobobs{0,0}=\sigma_z \otimes \openone$.
Similar arguments hold for the input $y_2=1$, which realizes a projective measurement of $\sigma_x$.
After identifying
\begin{equation}
    \begin{split}
       \bobproj{0,1}{+} &\equiv U^{\dagger}\qty[\frac{\left(\openone+\sigma_x \right)}{2}\otimes \openone ]U\\
       \bobproj{0,1}{-} &\equiv U^{\dagger}\qty[\frac{\left(\openone-\sigma_x \right)}{2}\otimes \openone ] U  \,,
    \end{split}
\end{equation}
we can conclude that $\bobobs{0,1}=\sigma_x \otimes\sigma_x$. 

In our protocol, reintroducing Alice, the state in $\mathcal{H}_A \otimes \mathcal{H}_{B'}$ is the Bell state $\ket*{\phi^+}$.
We note that since the state shared by Alice and the Bobs $\ket{\psi}= \ket*{\phi^+}\ket{\theta}$ is pure, Eve has no hope of gaining any information on the measurements outcomes, not even if $\theta \in \{0, \frac{\pi}{4}\}$, where we have said that she can have perfect correlations with one of $\bobobs{0}$, $\bobobs{0,1}$.

We now show other strategies valid for $\theta \in \{0, \frac{\pi}{4}\}$ that give Eve more power.
Let $\ket{\psi}_{ABE} \equiv \ket*{\phi^+}_{AB'}\ket*{\phi^+}_{B''E}$ be the global state including also Eve.
Operators $\aliceobs{0}, \aliceobs{1}, \bobobs{1}, \bobobs{0,0}$ are kept the same as above and since they act only on $\mathcal{H}_{AB'}$, on which the state is also kept the same, Alice and the Bobs' correlations when these operators are measured are unchanged.
However, for $\theta = 0$, Eve sets the Bobs' devices to measure in $\mathcal{H}_{B'}\otimes \mathcal{H}_{B''}$:
\begin{equation}
    \begin{split}
        \bobobs{0} &= \sigma_z \otimes \openone \\
        \bobobs{0,1} &= \openone \otimes \sigma_x
    \end{split}
\end{equation}
and measures $\sigma_x$ on her part of the state.
Instead, for $\theta = \frac{\pi}{4}$, Eve sets the Bobs' devices to measure:
\begin{equation}
    \begin{split}
        \bobobs{0} &= \openone \otimes \sigma_z \\
        \bobobs{0,1} &= \sigma_x \otimes \openone
    \end{split}
\end{equation}
and measures $\sigma_z$ on her part of the state.
In both cases, simple application of the Born rule shows that Alice and the Bobs' correlations are the same as above (also for the parts including $\bobobs{0}$ and $\bobobs{0,1}$), but Eve's outcome coincides with that of $\bobobs{0,1}$ (for $\theta = 0$) or $\bobobs{0}$ (for $\theta = \frac{\pi}{4}$).
This means that the guessing probability of the pair of outcomes of $\bobobs{0}$, $\bobobs{0,1}$ is $G \geq \frac12$, but considering the upper bound $G \leq \frac12$ proven in Section \ref{sec:app_minentropy}, we have $G = \frac12$.

\subsection{Experimental setup}
\label{sec:app_setup}
\begin{figure}
    \centering
    \includegraphics[width=.5\linewidth]{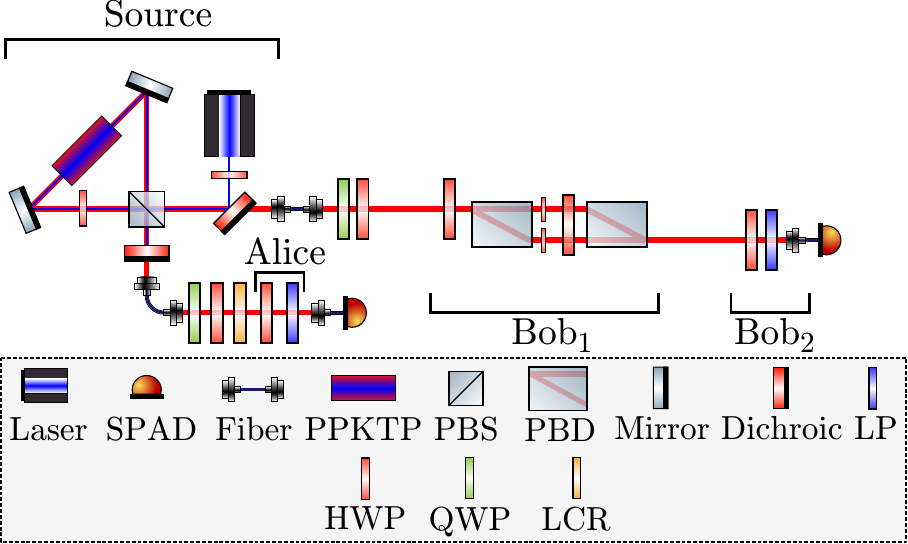}
    \caption{Experimental setup.}
    \label{fig:experimentalSetup}
\end{figure}

The experimental setup is depicted in Fig.\ \ref{fig:experimentalSetup}. 
We produce polarization-entangled photon pairs at approximately 810 nm \cite{Foletto2021} and send them to the two setups representing Alice and the Bobs using single mode fibers. There, their polarization is controlled in free space with quarter-wave plates (QWP) and half-wave plates (HWP) to approximately obtain the Bell state $\ket{\phi^+}=\qty(\ket{HH}+\ket{VV})/\sqrt{2}$.
Furthermore, a liquid crystal retarder is used on Alice's side to fine tune the relative phase between the two components.
Here, the $\ket{H}$ (horizontal) and $\ket{V}$ (vertical) states correspond to the states $\ket{0}$ and $\ket{1}$.

Alice uses a HWP and a linear polarizer (LP) to perform projective measurements.
Instead, on the other side, the setup is divided into Bob$_1$ and Bob$_2$: The first implements the weak measurement with a Mach-Zender interferometer (MZI) which couples polarization modes into spatial ones \cite{Foletto2021}.
The parameter $\theta$ of the coupling is controlled by setting the HWP shared between the two paths of the interferometer at $\pi/4 - \theta/2$. 
By selecting one outcome at a time with an external HWP, the MZI realizes the Kraus operators of the sequential protocol.
After the MZI, the photons are sent to Bob$_2$ where they encounter a projective measurement station, composed by an HWP and an LP.
Finally, at each side, the photons are coupled into single--mode fibers and then directed to single-photon avalanche diodes (SPADs) connected to a $1 \text{ ps}$-resolution timetagger that returns coincidence counts within a $\pm 0.55 \text{ ns}$ time window.

\end{document}

%% file: table.tex
\caption{Experimental results of the sequential CHSH experiment.
    Level 1+AB of the NPA hierarchy is used.
    Data retrieved with an exposure time of 100 s ($\sim 3\cdot 10^5$ coincidences).}
    \begin{tabular*}{\linewidth}{l @{\extracolsep{\fill}}ccccc}
        \toprule \toprule
         ID & $p$ & $c$ & $\theta$ & $H_{\mathrm{min}}$ \scriptsize (Model) & $H_{\mathrm{min}}$ \scriptsize (Experiment) \\
        & & & \scriptsize (rad) & \scriptsize (bits) & \scriptsize (bits)\\\midrule
        $1$ & $0.019$ & $0.017$ & $0.412$ & $0.82$ & $0.85 \pm 0.02$  \\
        $2$ & $0.016$ & $0.012$ & $0.436$ & $0.89$ & $0.86 \pm 0.01$ \\
        $3$ & $0.015$ & $0.012$ & $0.357$ & $0.90$ & $0.90 \pm 0.01$ \\
        \bottomrule\bottomrule
        \label{tab:seqchsh_tests_results}
    \end{tabular*}\\
        \begin{tabular*}{\linewidth}{l @{\extracolsep{\fill}}ccccccc}
        \toprule \toprule
         ID & $\expval*{S_1}$ \scriptsize (Model) & $\expval*{S_1}$ \scriptsize (Experiment) & $\expval*{S_2}$ \scriptsize (Model) & $\expval*{S_2}$ \scriptsize (Experiment) & $\expval*{S_c}$ \scriptsize (Model) & $\expval*{S_c}$ \scriptsize (Experiment)\\
        & & & & & &\\\midrule
        $1$ & $2.305$ & $2.292 \pm 0.002$ & $2.388$ & $2.421 \pm 0.003$ & $2.751$  & $2.738 \pm 0.003$\\
        $2$ & $2.270$ & $ 2.268\pm 0.002$ & $2.444$ & $ 2.433 \pm 0.003$ & $2.766$  & $ 2.760\pm 0.002$\\
        $3$ & $2.272$ & $2.432 \pm 0.002$ & $2.446$ & $2.250 \pm 0.003$ & $2.770$ & $2.778 \pm 0.002$\\
        \bottomrule\bottomrule
    \end{tabular*}
    \caption{Experimental results of the CHSH experiment.
    $\expval*{S}$ is the CHSH value and for the min-entropy the analytical bound is used \cite{Pironio2010}.
    Data retrieved with an exposure time of 100 s ($\sim 3\cdot 10^5$ coincidences).}
    \begin{tabular*}{\linewidth}{@{\extracolsep{\fill}}lcccc}
        \toprule \toprule
         ID &  $\expval*{S}$ \scriptsize (Experiment) & $H_{\mathrm{min}}$ \scriptsize (Model) & $H_{\mathrm{min}}$ \scriptsize (Experiment) \\
        & & \scriptsize (bits) & \scriptsize (bits)\\\midrule
        $1$ & $2.761\pm 0.003$ & $0.60$ & $0.61 \pm 0.01$ \\
        $2$ & $2.772 \pm 0.003$ & $0.63$ & $0.64 \pm 0.01 $ \\
        $3$& $2.797\pm 0.002$ & $0.64$ & $0.73 \pm 0.01$ \\
        \bottomrule\bottomrule
        \label{tab:chsh_tests_results}
    \end{tabular*}